\documentclass[10pt,journal,twocolumn]{IEEEtran}
\IEEEoverridecommandlockouts
\usepackage{verbatim}
\usepackage{amsmath}
\usepackage{amsmath,amsfonts}
\usepackage{algorithm}
\usepackage{algorithmic}
\usepackage{array}
\usepackage[caption=false,font=normalsize,labelfont=sf,textfont=sf]{subfig}
\usepackage{textcomp}
\usepackage{color}
\usepackage{stfloats}
\usepackage{booktabs}
\usepackage{url}
\usepackage{amsmath, amssymb}
\usepackage{amsmath}
\usepackage{verbatim}
\usepackage{graphicx}
\usepackage{amsthm}
\newtheorem{theorem}{Theorem}
\newtheorem{lemma}{Lemma}
\newtheorem{proposition}{Proposition}

\begin{document}

\title{\fontsize{24}{30}\selectfont Movable-Antenna Assisted Energy Minimization in UAV-Enabled Mobile Edge Computing Systems}

\author{Jiang Chen,
        Chunjie Wang,~\IEEEmembership{Graduate Student Member,~IEEE},
        Xuhui Zhang,~\IEEEmembership{Member,~IEEE},\\
        Yanyan Shen,~\IEEEmembership{Member,~IEEE},
        Kejiang Ye,~\IEEEmembership{Senior Member,~IEEE}, and
        Chengzhong Xu,~\IEEEmembership{Fellow,~IEEE}
 

\thanks{Jiang Chen is with Southern University of Science and Technology, Guangdong 518055, China,
with Shenzhen University of Advanced Technology, Shenzhen 518107, China,
and also with Shenzhen Institutes of Advanced Technology,
Chinese Academy of Sciences, Shenzhen 518055, China (e-mail: jiang.chen2@siat.ac.cn).}

\thanks{Xuhui Zhang is with the Shenzhen Future Network of Intelligence Institute, the School of Science and Engineering, and the Guangdong Provincial Key Laboratory of Future Networks of Intelligence, The Chinese University of Hong Kong, Shenzhen, Guangdong 518172, China. (e-mail: xu.hui.zhang@foxmail.com).}

\thanks{
Chunjie Wang,
Yanyan Shen, and Kejiang Ye are with Shenzhen Institutes of Advanced Technology, Chinese Academy of Sciences, Guangdong 518055, China (e-mail: cj.wang@siat.ac.cn; yy.shen@siat.ac.cn; kj.ye@siat.ac.cn).
}

\thanks{Chengzhong Xu is with the State Key Laboratory of IOTSC, Department of Computer Science, University of Macau, Macau, SAR, China (e-mail: czxu@um.edu.mo).}
\vspace{-18.5pt}
}

\maketitle

\begin{abstract}
Driven by the exponential growth of latency-sensitive applications, mobile edge computing (MEC) has emerged as a pivotal paradigm, yet mitigating its substantial energy consumption remains critical. This paper explores a movable-antenna (MA) assisted energy minimization scheme in an uncrewed aerial vehicle (UAV)-enabled MEC system, where a UAV equipped with an MA array serves as an edge server to process tasks offloaded from terrestrial consumer electronics (CE) devices. To minimize  the total system energy consumption, we jointly optimize computation resource allocation, CE transmit power, receive beamforming, and MA positions. To tackle the resulting non-convex problem with coupled variables, a robust alternating optimization algorithm based on the block coordinate descent  method is developed. The problem is iteratively decomposed into three subproblems. In particular, the subproblem of transmit power and receive beamforming is reformulated and optimized using the quadratic transform technique, while the MA array positions are optimized via the particle swarm optimization algorithm.  Numerical simulations verify that the proposed scheme achieves substantial energy savings over conventional benchmarks.
\end{abstract}


\section{Introduction}

\IEEEPARstart{T}{he} low-altitude economy (LAE), an emerging paradigm that employs low-altitude aerial vehicles as platforms over low-altitude wireless networks (LAWNs), has drawn increasing attention from researchers and industry practitioners \cite{10693833, wu2025low}. Within this LAE framework, we are witnessing a massive proliferation of smart consumer electronics (CE) devices, such as advanced smartphones, augmented reality/virtual reality headsets, and portable smart wearable devices \cite{11396947, 11471666}. These smart CE devices continuously generate explosive amounts of data to support computation-intensive and latency-sensitive consumer applications, including real-time aerial video streaming, immersive mobile gaming, and on-device artificial intelligence assistance. However, these resource-hungry applications impose a heavy burden on the physical constraints of the CE devices, highlighting two core consumer pain points, i.e., insufficient local processing capabilities and severe battery anxiety \cite{Zhang2023Learning}. Constrained by the coverage and processing capabilities of ground-based infrastructure, conventional wireless networks struggle to meet the strict quality-of-service (QoS) requirements of computation-intensive tasks from those CE devices in LAWNs, which further exacerbates the energy consumption of the CE devices \cite{yuan2025ground, 10980172}.

To address the aforementioned computational and battery limitations of CE devices, mobile edge computing (MEC), a paradigm that migrates computational resources to the network edge, provides CE devices with cloud-like services by deploying servers at the edge \cite{8387798, 10012694, wang2023incentive}. MEC systems can significantly reduce task processing latency and alleviate the computational burden on local CE devices. However, the traditional MEC system relies heavily on fixed ground stations, which results in limited service coverage and a lack of flexibility, failing to meet the dynamic demands of consumer mobility in LAWNs \cite{8956055, 10233771}.
To overcome these limitations, recent studies have proposed leveraging the three-dimensional (3D) mobility of uncrewed aerial vehicles (UAVs) to achieve on-demand deployment and wide-area coverage \cite{9273074, 9994654}. On one hand, by establishing line-of-sight (LoS) links between CE devices and UAVs, UAVs can act as relays to facilitate high-rate connections between CE devices and remote MEC servers \cite{8424236}. On the other hand, UAVs equipped with MEC servers can directly process offloaded computational tasks from CE devices, thereby further reducing system energy consumption and transmission latency \cite{9450021}. However, constrained by the limited on-board battery capacity, energy limitation remains a critical bottleneck in UAV enabled MEC systems \cite{10946517, liu2025computing}. Consequently, designing highly efficient resource allocation to satisfy the strict energy constraints of both UAVs and CE devices is crucial.

Furthermore, the evolution of consumer applications in LAWNs urgently demands highly efficient, low energy communication links between CE devices and UAV-enabled MEC servers. Conventional UAV-enabled MEC systems with fixed-position antennas (FPAs) struggle to adapt to dynamic wireless channels, often forcing CE devices to transmit at higher power levels and exacerbating battery anxiety \cite{11328802}. Fortunately, movable-antenna (MA) technology, also known as fluid antenna, has emerged as a key enabler for next-generation wireless networks \cite{10146274, 10286328}. Transcending the physical limitations of traditional FPAs, MAs flexibly adjust antenna positions within a given spatial region via mechanisms like liquid metal or reconfigurable pixels, thereby exploiting novel spatial degrees of freedom (DoFs) \cite{liu2025movable}. Through real-time channel reconfiguration, MAs significantly enhance signal quality and capacity \cite{zhu2024performance, 11374087}, directly reducing the transmit power required for task offloading and extending CE battery life. Crucially, MAs offer a cost-effective hardware solution for consumer-oriented UAVs by achieving superior beamforming gains with fewer radio frequency (RF) chains than large-scale arrays \cite{11543384}. This hardware efficiency makes MAs exceptionally suitable for resource-constrained consumer electronics, perfectly balancing communication performance with the strict cost and power limitations of aerial platforms.

Inspired by the immense potential of MA in enhancing consumer communication links, we investigate an MA-assisted UAV-enabled MEC system that is specifically designed to minimize the total energy consumption of the system, encompassing both the battery-constrained ground CE devices and the serving UAV. While existing research on energy minimization in UAV enabled MEC systems mainly focuses on the latency and energy consumption problem associated with task offloading and computation problem, and existing MA literature often optimizes antenna positions while overlooking the mechanical costs of antenna movement. To bridge this gap, this paper explicitly accounts for the latency and energy overheads incurred by MA position reconfiguration. The main contributions are summarized as follows:

\begin{itemize}
    \item  We investigate a novel MA-assisted UAV-enabled MEC system, where a UAV equipped with a linear MA array  hovers over ground CE devices, and provides computing services for the offloaded tasks of those devices. A comprehensive optimization problem is formulated to minimize the total energy consumption.
    \item  To address the non-convexity of the formulated problem, we apply the block coordinate descent (BCD) method to decompose it into three subproblems, and propose an efficient alternating optimization (AO) algorithm to iteratively optimize  computation resource allocation,  transmit power and receive beamforming, and  MA array  positions. Specifically, we employ the quadratic transform technique to reformulate the subproblem associated with the transmit power and receive beamforming.
    \item  Extensive simulations are conducted to evaluate the proposed scheme against established benchmark  with FPAs or random antenna position scheme. The findings reveal substantial gains in energy efficiency and operational stability, validating the efficacy of our proposed AO-based joint optimization scheme.
\end{itemize}

 \textit{Organizations:}
 The remainder of this paper is organized as follows. Section II reviews related works. Section III presents the system model and problem formulation. Section IV proposes an  AO-based algorithm and analyzes its complexity and convergence. Section V presents simulation results, and Section VI concludes the paper.

 \textit{Notations:} The notations are introduced below. $\mathbb{C}^{M \times N}$ denotes the $M \times N$ complex matrix. $\mathsf{j}$ represents the imaginary unit, satisfying $\mathsf{j}^2 = -1$. For complex  $z$ , $\Re\{z\} $  denotes its real part. For a generic matrix $\boldsymbol{G}$, $\boldsymbol{G}^{\mathsf{H}}$, $\boldsymbol{G}^{\mathsf{T}}$, $\mathrm{tr}(\boldsymbol{G})$denote its conjugate transpose, transpose, trace,  respectively. $\|\boldsymbol{a}\|$ denotes the Euclidean norm of vector $\boldsymbol{a}$.

\section{ {Related Works}}

\subsection{UAV-Enabled MEC Systems}
Benefiting from flexible deployment, high mobility, and cost-effectiveness, UAVs have been widely integrated into MEC systems, acting as aerial edge servers or relays \cite{9803857, 10606316}.
By proactively moving closer to users and optimizing network topology, UAV can establish controllable short-distance LoS links, facilitate spatial reuse of resources, and enable dynamic load migration \cite{9865231}. However, due to the limited on-board battery capacity, energy efficiency remains a critical bottleneck in UAV-enabled MEC systems. Consequently, extensive research has focused on minimizing energy consumption by strategically optimizing UAV positioning, flight trajectory, and resource allocation to simultaneously enhance service coverage and QoS.
Along this line, the authors in
\cite{liu2024energy} investigated the energy efficiency maximization problem in UAV-enabled MEC systems, subject to the energy constraints of both the UAVs and mobile users.
Beyond single-UAV scenarios, more sophisticated frameworks have been explored. The authors in \cite{qin2021joint} investigated a multi-UAV-assisted multi-access MEC system where single-antenna UAVs provided computing services to ground users, with the aim of minimizing the weighted sum energy consumption of both users and UAVs by jointly optimizing data allocation, transmit power, CPU frequency, and UAV trajectories.
In \cite{11122299}, the authors investigated a covert offloading framework for the UAV-enabled MEC system, aiming to minimize total energy consumption by jointly optimizing users  power splitting, computing resource allocation, and UAV deployment.
Moreover, the scope of UAV-enabled MEC system has been extended to emerging paradigm.
For instance, the authors in \cite{chen2025full} investigated a low-altitude UAV-enabled integrated sensing, communication, and computing system employing a multi-antenna UAV for simultaneous information transmission, sensing, and MEC services, where the total system energy consumption is minimized by jointly optimizing task allocation, computing resource allocation, transmit and receive beamforming.

Despite these extensive efforts, a fundamental limitation persists across the current literature, i.e., existing UAV-enabled MEC systems predominantly rely on FPA architectures. The rigid hardware configuration inherently constrains the system to a static or quasi-static channel environment, preventing the exploitation of micro-scale spatial DoFs that could be unlocked by dynamically adjusting antenna positions. Consequently, the inability to adapt the antenna layout to real-time channel conditions severely limits the potential for array gain optimization, thereby hindering further substantial improvements in overall energy efficiency.

\subsection{MA-Assisted Wireless Networks}
The MA technology has garnered significant attention in the next-generation wireless systems, owing to its unique capability to reconfigure antenna positions within confined regions. Such high reconfigurability not only endows the system with robustness in spatially highly correlated environments but also significantly enhances its capability to combat multipath fading \cite{11007274, 11374005}.
Inspired by these advantages, research has branched into several critical directions. 
In terms of secure and integrated systems, the authors in \cite{hu2024secure} proposed an MA-assisted framework for a single legitimate user against multiple eavesdroppers, maximizing the secrecy rate while guaranteeing reception quality.
Beyond security, MA has also been integrated into integrated sensing and communication (ISAC) systems.
For instance, \cite{10839251} designed a linear MA-assisted multi-user ISAC framework, where the beamforming and antenna positions were jointly optimized to maximize the weighted communication rate and sensing mutual information.
Furthermore, the high mobility of UAVs has been combined with MA to further exploit spatial diversity. In \cite{10654366}, a UAV communication system equipped with a linear MA array was investigated, maximizing the total achievable data rate for ground users through the joint design of transmit beamforming, UAV trajectory, and antenna positions.
Moreover, the application of MA has been extended to MEC systems to improve computation efficiency. The authors in \cite{10620306} proposed an MA-enabled wireless powered MEC system, where a fixed base stations (BS) equipped with MA simultaneously provided wireless energy transfer and task offloading services. Similarly, \cite{10528324}  investigated an MA-enabled MEC system, aiming to minimize the total system delay by jointly optimizing the offloading strategy and antenna positions at the BS.

However, the aforementioned MA-assisted MEC works primarily focus on terrestrial deployments where the MA arrays are integrated into fixed-position BS. While this configuration enhances local channel conditions, it inherently restricts service coverage to static areas and lacks the 3D mobility required to dynamically track the distribution of CE devices or bypass complex urban obstacles. Consequently, the unique potential of deploying MA arrays on UAVs, where the synergy between the UAV’s macro-scale trajectory optimization and the MA’s micro-scale position tuning could revolutionize MEC performance, remains largely unexplored in current research.

\section{System Model and Problem Formulation}
\subsection{System Model}
\begin{figure}
    \centering
    \includegraphics[width=0.64\linewidth]{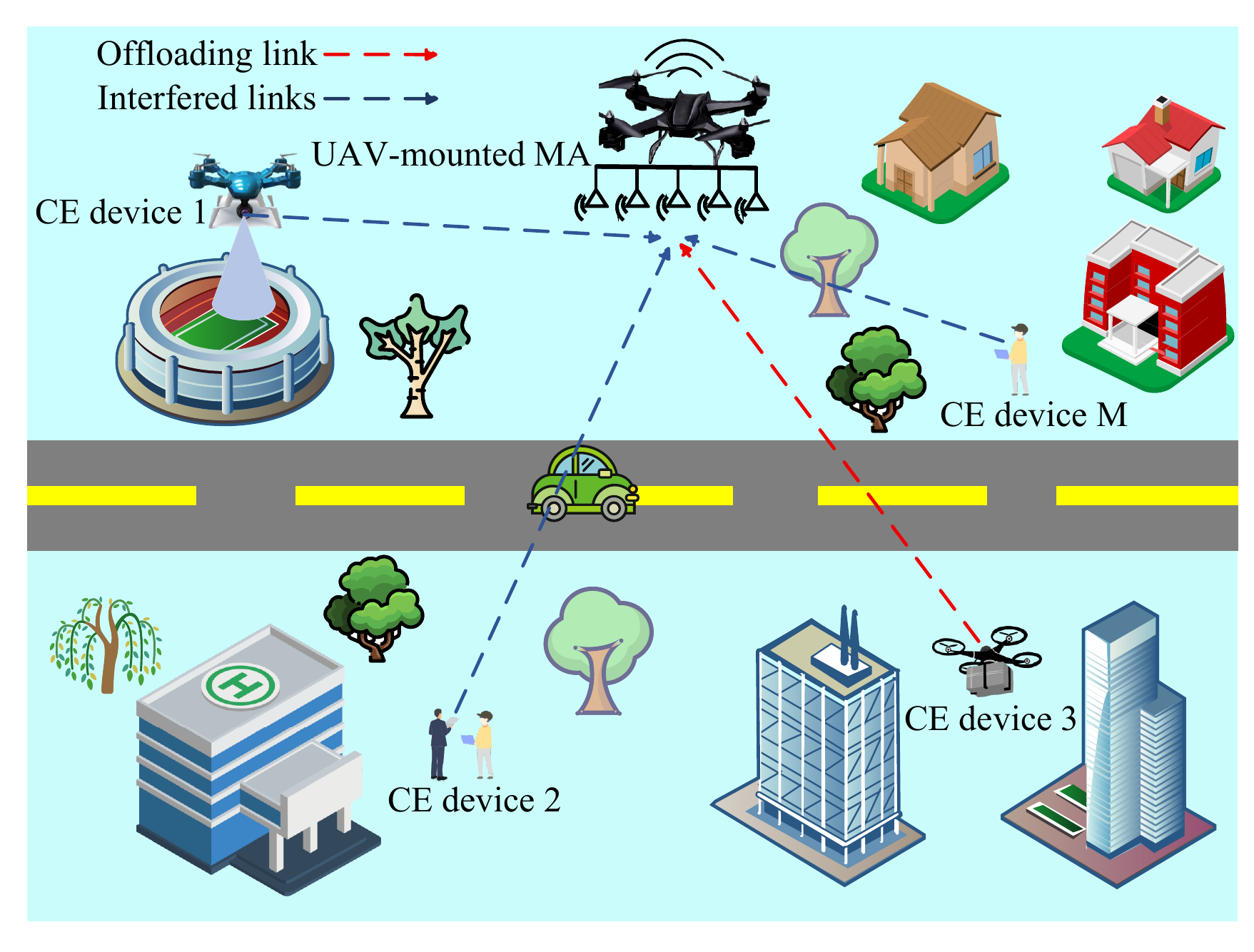}
        \caption{The system model of the UAV-enabled MEC system.}
        \label{fig:1}
\end{figure}

We consider an MA-enabled MEC system as depicted in Fig. \ref{fig:1}, which comprises a single UAV carrying $K$ MAs to serve $M$ single-antenna terrestrial CE devices. The system primarily facilitates uplink communications. We define the set of ground CE devices as $\mathcal{M} \triangleq \{ 1, \ldots, M\}$. The UAV hovers at a constant altitude $H$, with its horizontal coordinate in the current time slot given by $\mathbf{q} = [q_x, q_y]$, while the $m$-th user is located at $\boldsymbol{s}_m = [s_{m,x}, s_{m,y}]$. Let $\mathcal{K} \triangleq \{ 1, \dots, K \}$ represent the index set of the MAs. The positions of these $K$ antennas form an ordered sequence $\{x_k \}_{k=1}^K$ within the range $[0, L]$, satisfying $0 \leq x_1 \leq x_2 \leq \cdots \leq x_K \leq L$, where $L$ denotes the maximum moving span. To mitigate mutual coupling effects, a minimum spacing $d_{\text{min}}$ is enforced between any two elements, i.e., $||x_i - x_j|| \geq d_{\text{min}}$ for $\{i,j\} \in \mathcal{K},i\neq j$ \cite{10278220}.

A single time slot with duration $\tau$ is considered, which also represents the maximum tolerable latency. This duration encompasses antenna movement, uplink transmission, and edge computation. For the MA array, the steering vector is formulated as
\begin{equation}
\boldsymbol{a}_m(\mathbf{x}, \theta_m) = \left[ e^{j \frac{2\pi}{\lambda} x_1 \cos \theta_m}, \dots, e^{j \frac{2\pi}{\lambda} x_K \cos \theta_m} \right]^T,
\end{equation}
where $\lambda$ is the signal wavelength, $\theta_m = \arccos \frac{H}{\sqrt{\|\mathbf{q} - \mathbf{s}_m\|^2 + H^2}}$ signifies the angle of arrival (AoA) signals associated with the $m$-th CE devices.

 The existence of LoS paths is ensured by the operational height of the UAV. Thus, the channel response vector from user $m$ to the UAV is defined by
\begin{equation}
\boldsymbol{h}_m = \sqrt{ \frac{\rho}{|\mathbf{q} - \mathbf{s}_m|^2 + H^2}} \boldsymbol{a}_m,\label{eq1}
\end{equation}
where $\rho$ denotes the channel gain measured at a reference distance of 1 meter. {To enhance readability, the main notations and their descriptions are given in Table \ref{tnoa}.}

\begin{table}[!htbp] \footnotesize
	\centering
   
	\caption{{Notations and Descriptions}}
    
	\label{tnoa}
	\begin{tabular}{c>{\centering\arraybackslash}p{6.825cm}}	
  \toprule	
  Notations& Description\\	
  \midrule  %
    $k,K,\mathcal{K}$ & The index, number, and set of the MAs\\ 
  $m,M,\mathcal{M}$ & The index, number, and set of CE devices\\
  $\boldsymbol{s}_m$ &  The 2D position of the $m$-th CE device \\
  $\mathbf{q}$ &  The 2D position of the UAV \\
  $H$ & The altitude  of the UAV\\
  $x_k$ & The  position of the $k$-th antenna\\
 $ \theta_m$ & The AoA from $m$-th CE device to the UAV\\
  $\boldsymbol{a}_m$ & The receive array steering vector between the UAV and the $m$-th CE device\\
  $\rho$ &The channel power gain at a reference distance of 1m\\
$  \boldsymbol{h}_m$ & The channel vector from the $m$-th CE device to the UAV\\
  $\gamma_m$ & The received signal-to-interference-plus-noise ratio for the $m$-th CE device\\
  $B$ & The channel bandwidth\\
  $r_m$ & The communication rate between the UAV and the
$m$-th CE device \\
  
  $l$ & The task volume of  CE device\\
  $f^{\text{uav}}_{m}$ & The computation resource allocated to the $m$-th CE device
\\
 $f^{\text{uav}}_{\max}$ & The maximum available computation resources of  the UAV\\
  $c$ & The number of CPU cycles required to process one
bit of data\\
  $\xi$  &  The effective capacitance coefficient\\
  $p_m$ & The  transmit power of the $m$-th CE device\\
  $P_{{\max}}$ & The maximum transmit power of CE device\\
  $\sigma^2$ & The additive white Gaussian noise\\
$\bar{E}$ & The antenna moving energy cost per unit distance\\
  $x^0_k$ & The initial position of the $k$-the antenna\\
  $v$ & The antenna moving speed \\
  $L$ & The  maximum antenna moving span\\
		\bottomrule
	\end{tabular}
\end{table}

 \subsection{Communication Model}

The signal transmitted by the $m$-th CE device is defined as $z_m = \sqrt{p_m } \iota_m$. Where $p_m$ represents the transmit power of the CE device for task offloading, while $\iota_m$ denotes the information symbol, which satisfies $\mathbb{E}[|\iota_m|^2] = 1$. Accordingly, the signal received at the UAV after beamforming is modeled as 
\begin{equation}
 =\boldsymbol{w}^H_m \left (\sum_{m=1}^{M}\sqrt{p_m } \boldsymbol{h}_m \iota_m  + \boldsymbol{n} \right ), 
\label{eq3}
\end{equation}
where $\boldsymbol{n} \sim \mathcal{CN}(0, \sigma^2 \boldsymbol{I}_{K})$ is the additive white Gaussian noise (AWGN), and $\boldsymbol{w}_m$ is the receive beamforming vector. The signal-to-interference-plus-noise ratio (SINR) for CE device $m$ is formulated as 
\begin{equation}
 \gamma_m = \frac{p_m \left| \boldsymbol{w}^H_m \boldsymbol{h}_m \right|^2}{ \boldsymbol{w}^H_m \left (\sum\limits_{j=1, j\neq m}^{M} p_j   \boldsymbol{h}_j \boldsymbol{h}^H_j + \sigma^2 \right ) \boldsymbol{w}_m  }, \label{eq4}
\end{equation}
where $\|\boldsymbol{w}_m \|_2^2  \leq 1$. Therefore, the achievable uplink data rate for the $m$-th CE device is calculated as $r_m = B \log_2 \left( 1 + \gamma_m \right)$, with $B$ indicating the system bandwidth.

\subsection{Computing Model}
We assume  CE device generates an identical task volume $l$. We use a full offloading policy, meaning that CE devices do not conduct local computing. Instead, they offload their entire workload to the UAV for execution. 

The time required for the $m$-th CE device to offload its total  data to the UAV is calculated as 
\begin{align}
T^{\mathrm{tran}}_m = \frac{l}{ r_m }, \label{eqr}
\end{align}
and the transmission energy consumption can be calculated as
\begin{align}
E^{\mathrm{tran}}_m = \frac{l}{ r_m } p_m.
\end{align}
Subsequently, the task processing delay at the UAV is given by 
\begin{align}
T^{\mathrm{com}}_m = \frac{lc}{ f^{\mathrm{uav}}_m }, \label{eqf}
\end{align}
The corresponding computational energy consumption is 
$E^{\mathrm{com}}_m =\xi c\left( f^{\mathrm{uav}}_m \right)^2 l$,
where \(f^{\mathrm{uav}}_m\geq0\) denotes the UAV computing resources allocated to CE $m$. $c$ denotes the computation intensity, i.e., the number of CPU cycles required to process one bit of data, and $\xi$ is the effective capacitance coefficient determined by the chip architecture \cite{8849964, 11086503}. Additionally, the time consumed by the $k$-th antenna's movement is
\begin{align}
T^{\mathrm{mov}}_k &= \frac{|x_k - x^0_k|}{v},\label{eq10}
\end{align}
where $x^0_k \in [0, L] $ is the initial position and $v$ is the moving speed. The mechanical energy for this movement is determined by $E^{\mathrm{mov}}_k = \bar{E}\,{|x_k - x^0_k|} $,
where $\bar{E}$ denotes the energy cost per unit distance.

\subsection{Problem Formulation}
We aim to minimize the total energy consumption of the system by jointly optimizing the computation resource allocation $\boldsymbol{F}\stackrel{\triangle}{=}{\{f^{\mathrm{uav}}_m, \forall m\}}$, the transmit power $\boldsymbol{P} \stackrel{\triangle}{=}{\{p_m, \forall m\}}$, the receive  beamforming $\boldsymbol{W} \stackrel{\triangle}{=} {\{\boldsymbol{w}_m,\forall m\}}$, and the position of MA array $\boldsymbol{X}\stackrel{\triangle}{=}{\{x_k, \forall k\}}$. Then, the total energy consumption minimization problem can be formulated as
\begin{subequations}
\begin{flalign}
   (\textbf{P1}):\ &\min_{ \boldsymbol{F},\boldsymbol {P},\boldsymbol {W}, \boldsymbol{X}} \sum_{m=1}^{M} \left( \frac{l}{ r_m } p_m +E^{\mathrm{com}}_m \right )+\sum_{k=1}^{K} E^{\mathrm{mov}}_k \nonumber\\
{\rm{s.t.}} \quad
& \|\boldsymbol{w}_m\|^2  \leq 1, \,\forall m,\label{eq11a}\\
&0\leq p_m\leq P_{{\max}}, \, \forall m,\label{eq11b}\\
&0\leq f^{\text{uav}}_m, \sum_{m=1}^{M} f^{\text{uav}}_m\leq f^{\text{uav}}_{{\max}} , \, \forall m,\label{eq11c}\\
&T^\text{tran}_m+T^\text{com}_m +\max_{k \in K} \left\{ T^\text{mov}_k \right\}\leq \tau, \,\forall m,\label{eq11d}  \\
&\{x_k\}^K_{k=1} \in [0, L], \,\label{eq11e} \\
&x_i - x_j \geq d_{{\min}}, \, \{i, j\} \in \mathcal{K}, j \leq i\label{eq11f},
\end{flalign}
\end{subequations}
where \eqref{eq11a} represents the feasible region constraint of the receive beamforming. Constraint \eqref{eq11b} imposes the limit on CE device transmit power, with $P_{\max}$ representing its maximum power budget. Constraint \eqref{eq11c} governs the allocated  computation resource can not be negative, and the allocated computation resource to each CE cannot exceed the UAV's maximum computational resource threshold $f^{\text{uav}}_{{\max}}$. Constraint \eqref{eq11d} corresponds to the system latency requirement. Specifically, $T_m^{\mathrm{tran}}$ depends on the achievable rate $r_m$ as shown in equation \eqref{eqr}, and $r_m$ is jointly determined by the transmit power $\boldsymbol{P}$, the receive beamforming  $\boldsymbol{W}$, and the MA positions $\boldsymbol{X}$ as shown in equation \eqref{eq4}. The computation delay $T_m^{\mathrm{com}}$ depends on the computation resource allocation $\boldsymbol{F}$ as given in equation \eqref{eqf}, while as equation \eqref{eq10} shows the antenna movement delay $T_k^{\mathrm{mov}}$ depends on the MA positions $\boldsymbol{X}$.  Furthermore, \eqref{eq11e} and \eqref{eq11f} define the spatial restrictions for the MA array.

Problem $\textbf{P1}$ is a non-convex optimization problem, primarily due to the intricate coupling among $\boldsymbol{P}$, $\boldsymbol{W}$, and $\boldsymbol{X}$ in the rate expression $r_m$. Furthermore, the high non-linearity of the steering vector $\boldsymbol{a}_m(\mathbf{x}, \theta_m)$ with respect to the MA positions $\boldsymbol{X}$ renders the transmission rate non-concave. These issues, compounded by the fractional objective function and the latency constraint \eqref{eq11d}, result in an intractable, non-convex feasible region that cannot be solved directly.

\section{Proposed Solution}

 Due to the non-convex nature of problem $(\textbf{P1})$, finding its global optimum directly is highly intractable. To address this, we propose an  AO algorithm based on the BCD method to efficiently seek a high-quality suboptimal solution. The proposed algorithm iteratively updates three blocks of variables: computation resource allocation, transmit power and receive beamforming, and MA array positions. Furthermore, since the subproblem associated with the transmit power and receive beamforming exhibits a complex fractional structure, we apply the quadratic transform technique to decouple the numerator and denominator terms.

\subsection{Computation Resource Allocation}
Given  $\{\boldsymbol {P},\boldsymbol {W}, \boldsymbol{X}\}$, and removing the irrelevant terms, the problem $(\textbf{P2})$  of optimizing the  computation resource allocation  $\boldsymbol{F}$ can be formulated as
\begin{align}
   (\textbf{P2}):\ &\min_{
   \boldsymbol{F}}  \sum_{m=1}^{M} \xi c\left( f^{\mathrm{uav}}_m \right)^2 l  \nonumber\\
{\rm{s.t.}} \quad     
&\eqref{eq11c},\eqref{eq11d} \nonumber
\end{align}
where the convex nature of $(f_m^{\mathrm{uav}})^2$ and $\frac{1}{f_m^{\mathrm{uav}}}$ with respect to $f_m^{\mathrm{uav}}$ implies that the objective function and constraint \eqref{eq11c} is convex, while constraint \eqref{eq11d} is linear. Consequently, problem (\textbf{P2}) is a convex optimization problem amenable to efficient solution by the CVX solver.

\subsection{Transmit Power and Receive Beamforming Optimization}
Given  $\{\boldsymbol{F},\boldsymbol{X}\}$, and removing the irrelevant terms, the subproblem of optimizing transmit power $\boldsymbol{P}$ and  receive  beamforming $\boldsymbol{W}$ is formulated as follows
\begin{align}
    (\textbf{P3}):\ &\min_{\boldsymbol{W},
    \boldsymbol{P}} \sum_{m=1}^{M} \frac{l}{ r_m } p_m \nonumber\\
{\rm{s.t.}} \quad 
&\eqref{eq11a},\eqref{eq11b},\eqref{eq11d}.\nonumber
\end{align}

As can be readily observed, the optimization variables $\boldsymbol{w}_m$ and $p_m$ are tightly coupled in fractional structures across the objective function of problem $(\textbf{P3})$, the SNR formula in \eqref{eq4}, and the data rate $ r_m$ expression. To make the optimization problem more mathematically manageable, we employ the quadratic transform technique to restructure the problem, as outlined in the following lemma.

\begin{lemma} \label{qtt}
Consider \( M \times 1 \) pairs of non-negative functions \( A_{m}(\mathbf{x}) : \mathbb{R}^d \to \mathbb{R}^+ \) and positive functions \( B_{m}(\mathbf{x}) : \mathbb{R}^d \to \mathbb{R}^{++} \), the optimization problem minimizing the sum-of-ratios
\begin{align}
    \min_{\{\mathbf{x}\}} & \quad \frac{1}{M} \sum_{m=1}^{M}  \frac{A_{m}(\mathbf{x})}{B_{m}(\mathbf{x})}  \nonumber\\
    \text{s.t.} 
    & \quad \mathbf{x} \in \mathcal{X}, 
\end{align}
can be equivalently rewritten as
\begin{align}
    \min_{\{\mathbf{x},\boldsymbol{\alpha}\}} & \quad \frac{1}{M}  \sum_{m=1}^{M} \left( 2 \alpha_{m} \sqrt{A_{m}(\mathbf{x})} -  \alpha_{m}^2 B_{m}(\mathbf{x}) \right) \nonumber\\
    \text{s.t.} & \quad \mathbf{x} \in \mathcal{X}, 
\end{align}
\textit{where $\mathbf{x}$ represents the decision variable constrained by the feasible set $\mathcal{X}$, and $\boldsymbol{\alpha}$ denotes the auxiliary variable vector of dimension $M \times 1$, with $ \alpha_{m}$ being the $m$-th element in $\boldsymbol{\alpha}$.}
\end{lemma}
\begin{proof}
    Please refer to \cite{8314727}.
\end{proof}
According to Lemma \ref{qtt}, problem $(\textbf{P3})$ is reformulated as
\begin{subequations}
\begin{flalign}
   (\textbf{P4}):\ &\min_{ \boldsymbol {P},\boldsymbol {W}, \boldsymbol{\alpha}} \sum_{m=1}^{M}  2 \alpha_m\sqrt{lp_m}- \alpha_m^{2}{ r}_m\nonumber\\
{\rm{s.t.}} \quad
&\eqref{eq11a},\eqref{eq11b},\eqref{eq11d}\nonumber.
\end{flalign}
\end{subequations}
where the variable  $ \boldsymbol{\boldsymbol{ \alpha}} = \{ \alpha_m, \forall m\}$ is defined as an $M \times 1$ vector. $ \alpha_{m}$ is alternately updated as follows:
\begin{equation}
    \alpha_m = \frac{\sqrt{lp_m}}{r_m},\forall m, \label{eq1120} 
\end{equation}
Given $\boldsymbol{ \alpha}$, variable block $\{\boldsymbol{P}, \boldsymbol{W}\}$ is optimized. The convergence of this alternating optimization process is theoretically guaranteed, as the objective function value monotonically decreases and is bounded below. Specifically, the algorithm initializes $\boldsymbol{ \alpha}$ with small constants alongside a feasible $\{\boldsymbol{P}, \boldsymbol{W}\}$, and then iteratively updates $\{\boldsymbol{P}, \boldsymbol{W}\}$ and $\boldsymbol{ \alpha}$ until convergence, yielding the optimal solutions $\{\boldsymbol{P}^*, \boldsymbol{W}^*\}$ and $\boldsymbol{ \alpha}^*$.

However, the direct solution of problem $(\textbf{P4})$ is hindered by its highly non-convex objective function.  The data rate $ r_m$ exhibits non-convexity regarding the variables $ p_m$ and $ \boldsymbol{w}_m$. To address this intractability, we leverage the weighted minimum mean square error (WMMSE) technique to transform $ r_m$ into an equivalent tractable form, detailed in \eqref{eq33},  where $\boldsymbol{\mathcal{G}} {=} {\{\mathcal{G}_m,\forall m\}}$ and $\boldsymbol{\mathcal{I}} {=} {\{\mathcal{I}_m,\forall m\}}$ are the introduced auxiliary variable sets \cite{5962843}. Consequently, the problem $(\textbf{P4})$ is converted to the following approximate form:
\begin{align}
    (\textbf{P4.1}):\ &\min_{\boldsymbol{W},
    \boldsymbol{P},\boldsymbol{\mathcal{I}},\boldsymbol{\mathcal{G}} } \sum_{m=1}^{M} 2\alpha_m\sqrt{lp_m}-\alpha_m^{2}\tilde{ r}_m  \nonumber\\
{\rm{s.t.}} \quad
&\mathcal{G}_m\geq 0,\, \forall m,\\
&\eqref{eq11a},\eqref{eq11b},\eqref{eq11d} \nonumber.
\end{align}

Subsequently, the problem$(\textbf{P4.1})$ is addressed by employing the BCD method.
\paragraph{Updating auxiliary variables}
Based on the WMMSE framework, the optimal updates for the auxiliary variables $\mathcal{G}_m$ and $\mathcal{I}_m$ are derived as follows
    \begin{align}  
    \mathcal{I}_m^{\text{opt}} &= \frac{\sqrt{p_m}  \boldsymbol{w}^H_m \boldsymbol{h}_m }{\sum\limits_{j=1}^{M} p_j \left| \boldsymbol{w}^H_m \boldsymbol{h}_j \right|^2 + \|\boldsymbol{w}^H_m\|^{2}_2 \sigma^2},\label{var1}\\
   \mathcal{G}_m^{\text{opt}} &= 1 + \frac{p_m \left| \boldsymbol{w}^H_m \boldsymbol{h}_m \right|^2}{\sum\limits_{j=1, j\neq m}^{M} p_j \left| \boldsymbol{w}^H_m \boldsymbol{h}_j \right|^2 + \|\boldsymbol{w}^H_m\|^{2}_2 \sigma^2}.
   \label{var2}
\end{align}

\paragraph{Updating the transmit power}
The original problem is transformed into a new form
\begin{align}
    (\textbf{P4.2}):\ &\min_{
    \boldsymbol{P}^{\mathrm{new}}} \sum_{m=1}^{M} 2\alpha_m\sqrt{l}p^{\mathrm{new}}_m-\alpha_m^{2}(\tilde{ r}_m)^{\mathrm{new}}   \nonumber\\
 {\rm{s.t.}} \quad
&0\leq p^{\mathrm{new}}_m\leq  \sqrt{P_{\text{max}}}, \, \forall m,\label{22a}\\
&\frac{l_m}{(\tau-\max_{k \in K} \left\{ T^\text{mov}_k \right\}-\frac{l_mc}{ f^{\mathrm{uav}}_m } )  }\leq (\tilde{ r}_m)^{\mathrm{new}}, \,\forall m\label{22b}.
\end{align}
Where we define the auxiliary power variable set $\boldsymbol{P}^{\mathrm{new}} \triangleq \{ p^{\mathrm{new}}_m , \forall m \}$ by setting $p^{\mathrm{new}}_m = \sqrt{p_m}$, with $(\tilde{ r}_m)^{\mathrm{new}} $ given in Equation \eqref{eq34}. As problem $(\textbf{P4.2})$ falls into the category of convex optimization, it is computationally tractable using the CVX toolbox. Moreover, by introducing dual variables associated with the constraints \eqref{22a} and \eqref{22b}, we can derive a Lagrangian dual function and thus an optimal solution can be derived. Concrete information is shown in Theorem \ref{the1}.

\begin{theorem} \label{the1}
For problem (\textbf{P4.2}), the transmit power can be expressed as
\begin{align}
&\left(p_{m}^{\mathrm{new}}\right)^{\mathrm{opt}} =\nonumber\\
&\frac{2\alpha_m\sqrt{l}- \mathcal{Z}_{m}+2B \mathcal{G}_m \Re{\{\mathcal{I}_m^* \boldsymbol{w}^H_m \boldsymbol{h}_m\} }\left(\alpha_m^{2}-\varpi_m\right)}{2B\mathcal{G}_m|\mathcal{I}_m|^2\sum_{j=1}^{M}\left| \boldsymbol{w}^H_m \boldsymbol{h}_j \right|^2\left(\alpha_m^{2}-\varpi_m\right)},
\label{eq:33}
\end{align}
where $\{\varpi_m, \mathcal{Z}_{m} \}$ are the dual variables associated with the corresponding constraints \eqref{22a} and \eqref{22b}.
\end{theorem}

\begin{proof}
Please refer to Appendix \ref{App1}.
\end{proof}

\paragraph{Updating the receive beamforming}
\begin{align}
    (\textbf{P4.3}):\ &\min_{
    \boldsymbol{W} } 2\alpha_m\sqrt{lp_m}-\alpha_m^{2}\tilde{ r}_m \nonumber\\
 {\rm{s.t.}} \quad
& \frac{l}{\left(\tau-\max_{k \in K} \left\{ T^\text{mov}_k \right\}-\frac{lc}{ f^{\mathrm{uav}}_m }\right )}\leq \tilde{ r}_m, \,\forall m,\label{eq26a}\\
&\eqref{eq11a} \nonumber.
\end{align}

Since problem $\textbf{P4.3}$ is convex, it can  be  solved using standard convex optimization tools such as CVX. Moreover, by introducing dual variables associated with the constraints \eqref{eq11a}  and \eqref{eq26a}, we can derive a Lagrangian dual function and thus an optimal solution can be derived. Concrete information is shown in Theorem  \ref{the2}.
\begin{theorem} \label{the2}
For problem (\textbf{P4.3}), the  optimal receive  beamforming can be expressed as
\begin{align} 
&\left(\boldsymbol{w}_m\right)^{\mathrm{opt}} =\nonumber\\
&\frac{\left(\alpha_m^{2}-\beta_m\right)B\mathcal{G}_m\sqrt{p_m }\boldsymbol{h}_m\mathcal{I}_m }{B\mathcal{G}_m|\mathcal{I}^*_m|^2  \left(\sum\limits_{j=1}^{ M} p_j  \boldsymbol{h}^*_m\boldsymbol{h}^{\text{T}}_m+\sigma^2\boldsymbol{I}\right)\left(\beta_m-\alpha_m^{2}\right)+\phi_{m}\boldsymbol{I}}.
\end{align}
where $\{\beta_m, \phi_{m} \}$ are the dual variables associated with the corresponding constraints \eqref{eq11a} and \eqref{eq26a}.
\end{theorem}

\begin{proof}
Please refer to Appendix \ref{App2}.
\end{proof}

The detailed BCD procedure to solve problem (\textbf{P3}) is summarized in Algorithm \ref{algorithm1}.

\begin{algorithm}[t]\footnotesize
\caption{BCD-based Algorithm for Solving Problem (\textbf{P3}).}
\label{algorithm1}  
\begin{algorithmic}[1]
\REQUIRE
An initial feasible solution $\{\mathbf{P}^{(0)},\mathbf{W}^{(0)}\}$, 
\STATE
        \textbf{Initialize:} iteration index $ j = 0 $, maximum iteration number $j_{\rm{max}}$, accuracy threshold $ \varepsilon > 0 $;

\STATE \textbf{Repeat:}
\STATE Update the quadratic transform variable $y$ according to \eqref{eq1120};
\STATE Update the variables $\mathcal{I}$ and $\mathcal{G}$ according to \eqref{var1} and \eqref{var2};
\STATE Update the transmit power by solving problem \textbf{P4.2};
\STATE Update the receive beamforming by solving problem \textbf{P4.3};
\STATE Update the objective function value of problem \textbf{P3};
\STATE Update $j = j + 1$;
\STATE \textbf{Until:} the decrease of the value of the objective function between two adjacent iterations is smaller than $\varepsilon$ or $j > j_{\max}$;
\ENSURE $\{\mathbf{P}^{(j)},\mathbf{W}^{(j)}\}$.
\end{algorithmic}
\end{algorithm}

\begin{figure*}[!t]
\vspace*{-0.8\baselineskip}
\footnotesize

\begin{equation}
\begin{aligned}
r_m &= B \left(\log_2 \left( \mathcal{G}_m \right)- \mathcal{G}_m \left(\sqrt{p_m } \boldsymbol{w}^H_m \boldsymbol{h}_m\right) \left(\sum\limits_{j=1}^{ M}  p_j \left| \boldsymbol{w}^H_m \boldsymbol{h}_j \right|^2 + ||\boldsymbol{w}^H_m||^{2}_2\sigma^2\right)^{-1}+1\right)\\
&= B \left(\log_2 \left( \mathcal{G}_m \right)- \mathcal{G}_m \left(1-2\Re{\{\mathcal{I}_m^*\sqrt{p_m } \boldsymbol{w}^H_m \boldsymbol{h}_m\} }+|\mathcal{I}_m|^2 \left( \sum\limits_{j=1}^{ M}  p_j \left| \boldsymbol{w}^H_m \boldsymbol{h}_j \right|^2 + ||\boldsymbol{w}^H_m||^{2}_2\sigma^2\right)\right)+1\right)\\
&=\tilde{ r}_m \label{eq33}
\end{aligned}
\end{equation}
\hrulefill
\end{figure*}

\begin{figure*}[!t]
\vspace*{-0.8\baselineskip}
\footnotesize

\begin{equation}
\begin{aligned} 
&\left(\tilde{ r}_m\right)^{\mathrm{new}} = B \left(\log_2 \left( \mathcal{G}_m \right)- \mathcal{G}_m \left(1-2\Re{\{\mathcal{I}_m^*p^{\mathrm{new}}_m \boldsymbol{w}^H_m \boldsymbol{h}_m\} }+|\mathcal{I}_m|^2 \left( \sum\limits_{j=1}^{ M}  \left(p^{\mathrm{new}}_j\right)^2 \left| \boldsymbol{w}^H_m \boldsymbol{h}_j \right|^2 + ||\boldsymbol{w}^H_m||^{2}_2\sigma^2\right)\right)+1\right)
 \label{eq34}
\end{aligned}
\end{equation}
\hrulefill
\end{figure*}

\subsection{MA Positions Optimization}
Given  $\{\boldsymbol{F},\boldsymbol {P},\boldsymbol {W} \}$, the problem $(\textbf{P1})$  of optimizing the position of MA array $\boldsymbol{X}$ can be formulated as
\begin{equation*}
\begin{aligned}
   (\textbf{P5}):\ &\min_{
   \boldsymbol {X}} \sum_{m=1}^{M} ( E^{\mathrm{tran}}_m +E^{\mathrm{com}}_m)+\sum_{k=1}^{K} E^{\mathrm{mov}}_k  \\
{\rm{s.t.}} \quad
  &  \eqref{eq11d}, \eqref{eq11e},\eqref{eq11f}.
\end{aligned}
\end{equation*}
where problem $(\textbf{P5})$ demonstrates strong non-convexity alongside a high-dimensional solution domain for antenna location. To mitigate these issues, PSO is applied as the optimization strategy.

 Specifically, the process begins by generating a swarm of $U$ random feasible candidates for problem $(\textbf{P5})$. We define this solution set as $\mathcal{P} = \{\mathcal{P}_{(i)}^1, \mathcal{P}_{(i)}^2, \cdots, \mathcal{P}_{(i)}^{U}\}$, where the coordinate vector of the $u$-th particle at iteration $i$ is expressed as $\mathcal{P}_{(i)}^u = [x_{1,(i)}^u, \cdots, x_{K,(i)}^u]^T$. During the optimization, each particle dynamically adjusts its trajectory by balancing its own historical findings with the swarm's collective intelligence. This collaborative search mechanism effectively guides the population toward the global optimum $\mathcal{P}^* = [x_{1}^*, \cdots, x_{K}^*]^T$.

The local optimal solution $\mathcal{P}_{(i)}^{u*}$ of the $u$-th particle during the $i$-th iteration and the global optimal solution $\mathcal{P}_{(i)}^{*}$ from the swarm are given by
   \begin{align}  
    \mathcal{P}_{(i)}^{u*} &= \arg \min_{\mathcal{P}^u \in \{\mathcal{P}_{(1)}^{u}, \cdots, \mathcal{P}_{(i)}^{u}\}} \sum_{m=1}^{M} \left( E^{\mathrm{tran}}_m +E^{\mathrm{com}}_m\right) +\sum_{k=1}^{K} E^{\mathrm{mov}}_k,\\
    \mathcal{P}_{(i)}^{*} &= \arg \min_{\mathcal{P}^u \in \{\mathcal{P}_{(i)}^{1*}, \cdots, \mathcal{P}_{(i)}^{U*}\}} \sum_{m=1}^{M} \left( E^{\mathrm{tran}}_m +E^{\mathrm{com}}_m\right)  +\sum_{k=1}^{K} E^{\mathrm{mov}}_k.
\end{align}  
The update velocity of the $u$-th particle at the $i$-th iteration is given by
\begin{equation}
\mathcal{V}_{(i+1)}^u = \varsigma \mathcal{V}_{(i)}^u + \varrho \mathcal{S}_{1} \left(\mathcal{P}_{(i)}^{u*} - \mathcal{P}_{(i)}^u\right) +\mu \mathcal{S}_{2} \left(\mathcal{P}_{(i)}^{*} - \mathcal{P}_{(i)}^{u}\right),\label{pvel}
\end{equation}
where  $\varrho$ and $\mu$ denote the cognitive parameter and social parameter, respectively, and both of them are positive constants. $\mathcal{S}_{1}$ and $\mathcal{S}_{2} $ are random parameters uniformly distributed within $[0, 1]$,       $\varsigma$ represents the inertia weight, which can be expressed as
\begin{equation}
    \varsigma = \left( \varsigma_{\text{max}} - \frac{ \left(\varsigma_{\text{max}} - \varsigma_{\text{min}}\right)i}{i_{\text{max}}} \right),\label{inertia_weight}
\end{equation}
where $\varsigma_{\text{max}}$ and $\varsigma_{\text{min}}$ denote the maximum and minimum values of $\varsigma$, respectively, and $i_{\text{max}}$ represents the maximum number of iterations. Then, the solution is updated by
\begin{equation}
   \mathcal{P}_{(i+1)}^u = \mathcal{P}_{(i)}^u + \Gamma \mathcal{V}_{(i+1)}^u,\label{ppos}
\end{equation}
where $\Gamma$ is a constant for update control. 

Next, we address the constraints associated with problem \textbf{P5}. The boundary limits defined in \eqref{eq11e} are enforced by projecting the position coordinates of particles onto the valid interval, thereby ensuring feasibility. In parallel, constraints \eqref{eq11d} and \eqref{eq11f} are satisfied by incorporating an adaptive penalty term into the fitness evaluation, which steers the optimization toward valid regions, i.e.,
\begin{equation}
\begin{aligned}
\delta(\mathcal{P}_{(i)}^{u}) &= \sum_{m=1}^{M} ( E^{\mathrm{tran}}_m (\mathcal{P}_{(i)}^{u})+E^{\mathrm{com}}_m(\mathcal{P}_{(i)}^{u})+\sum_{k=1}^{K} E^{\mathrm{mov}}_k(\mathcal{P}_{(i)}^{u}) \\
&\quad\ + \psi_1 |\Psi_1(\mathcal{P}_{(i)}^{u})| + \psi_2 |\Psi_2(\mathcal{P}_{(i)}^{u})|,\label{PSOfitness}
\end{aligned}
\end{equation}
the fitness value of a particle, represented by $\delta(\mathcal{P}_{(i)}^{u})$, incorporates penalty coefficients $\psi_1$ and $\psi_2$ to address constraint breaches. Specifically, the penalty functions $\Psi_1(\mathcal{P}_{(i)}^{u})$ and $\Psi_2(\mathcal{P}_{(i)}^{u})$ correspond to the constraints defined in \eqref{eq11d} and \eqref{eq11f}, respectively. By iteratively assessing the fitness, both the individual historical best and the swarm's global best positions are continuously updated until the algorithm converges. 
For simplify,
the detailed procedure of the PSO-based algorithm
is omitted here. Please refer to the Algorithm 3 in \cite{11240557}.
\begin{figure}
    \centering
    \includegraphics[width=0.85\linewidth]{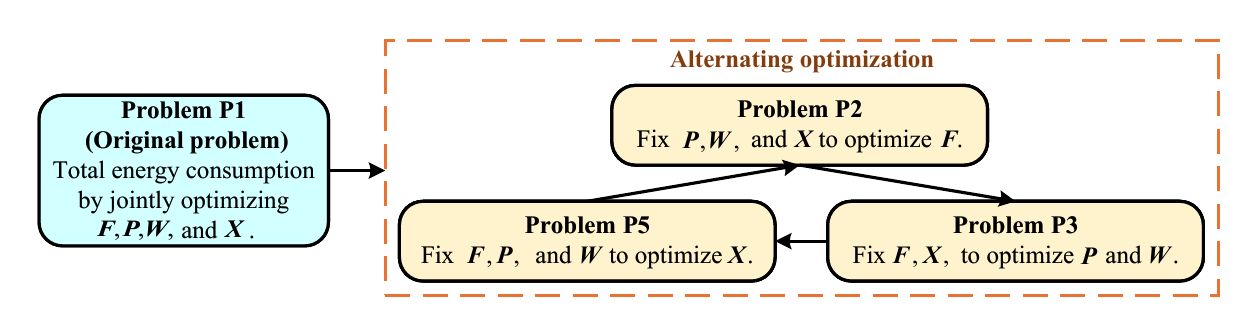}
        \caption{ The sequential steps of Algorithm \ref{algorithmAO}.}
        \label{AO11}
\end{figure}

\subsection{Convergence and Complexity Analysis}
{ We employ an AO framework to address problem \textbf{P1}, the rigorous step-by-step details are outlined in Algorithm \ref{algorithmAO}. To enhance readability and provide an intuitive overview of the iterative architecture, the macroscopic workflow of the proposed AO procedure is visually illustrated in Fig. \ref{AO11}.  Fig. \ref{AO11} offers a clean, high-level understanding of the problem transformation and the iterative updates among the decoupled subproblems, Algorithm \ref{algorithmAO} provides the specific initialization steps, convergence criteria, and exact mathematical parameters necessary to guarantee full reproducibility. The convergence behavior of this algorithm is subsequently examined in the following proposition.}

\begin{proposition} \label{prop1}
    Algorithm \ref{algorithmAO} is guaranteed to converge.
    \label{pconvergence}
\end{proposition}
\begin{proof}
    Please refer to the Appendix \ref{App3}.
\end{proof}

Having established the convergence properties, we  evaluate the computational complexity of the proposed scheme. In each iteration of Algorithm \ref{algorithmAO}, the interior-point method is employed to solve three convex subproblems: computation resource allocation,  transmit power control, and receive beamforming. The computational costs associated with these subproblems are given by $\tilde{\mathcal{O}}_1 = \mathcal{O}((M)^{3.5} \log(\epsilon^{-1})) $, $ \tilde{\mathcal{O}}_2 = \mathcal{O}((M)^{3.5} \log(\epsilon^{-1}))$, and $ \tilde{\mathcal{O}}_3 = \mathcal{O}((MK^2)^{3.5} \log(\epsilon^{-1})) $, respectively, where $\epsilon$ denotes the solution accuracy. Meanwhile, the MA positions are updated via the PSO algorithm, which incurs a complexity of $\tilde{\mathcal{O}}_4 =\mathcal{O}(i_{\max}U(K+\log(U)))$. Consequently, the overall computational complexity of Algorithm \ref{algorithmAO} can be summarized as $ \mathcal{O}\left(l_1\left(\tilde{\mathcal{O}}_1 +l_2(\tilde{\mathcal{O}}_2 + \tilde{\mathcal{O}}_3) + \tilde{\mathcal{O}}_4 \right)\right)$, where $l_1$ represents the number of iterations required for convergence, $l_2$ denotes the resultant iteration number of Algorithm \ref{algorithm1}.

\section{Numerical Results}
This section presents numerical results to corroborate the efficacy of our proposed scheme. Simulations are conducted in a $400\text{m} \times 400\text{m}$ region serving $M = 4$ CE devices. The UAV utilizes an array of $K = 6$ antennas positioned at a height of $H=20$ m, with an element spacing of $d_{\text{min}} = \lambda /2$ and a movable range of $L= 4\lambda$. The data load per user is $l= 1 \times 10^7 \text{bits}$, processed at $c = 1000$ cycles/bit. The maximum tolerable latency is $\tau = 50 \text{s}$.  The maximum transmit power is $P_{\max} = 1\text{W}$. Channel specifications include a reference gain of $\rho = -30\text{dB}$, noise power of $\sigma^2 = -80\text{dBm}$, and bandwidth of $B = 1.45\text{MHz}$. For antenna mobility, the speed is $v = 0.2 \text{m/s}$ and energy consumption is $\bar{E}= 0.175 \text{J/m}$. The maximum available computation resources of  the UAV is $f^{\text{uav}}_{\max}=5\text{GHz}$. The PSO algorithm is configured with a population of $U = 200$, $i_{\max} = 200$ iterations, learning factors $\varrho=\mu = 1.5$, and inertia weights decreasing from $\varsigma_{\max} = 0.9$ to $\varsigma_{\min} = 0.4$. Penalty terms are set to $\psi_1 = \psi_2 = 1000$. We benchmark our proposed scheme against: 1) \textit{Fixed-position antenna array} (FPA), which fixes the array location; and 2) \textit{Random-position antenna array} (RPA), which generates random positions satisfying constraints \eqref{eq11e} and \eqref{eq11f}.

Fig. \ref{fig2} depicts the convergence trends of the proposed scheme. By simultaneously optimizing computation resource allocation, transmit power and receive beamforming, and MA array positions, the MA-assisted scheme rapid convergence with minimal total energy consumption. Conversely, the FPA scheme exhibits slower convergence and higher energy costs due to its rigidity, whereas the RPA scheme results in erratic and inefficient performance, failing to capitalize on channel variations.
\begin{figure}
    \centering
    \includegraphics[width=0.605\linewidth]{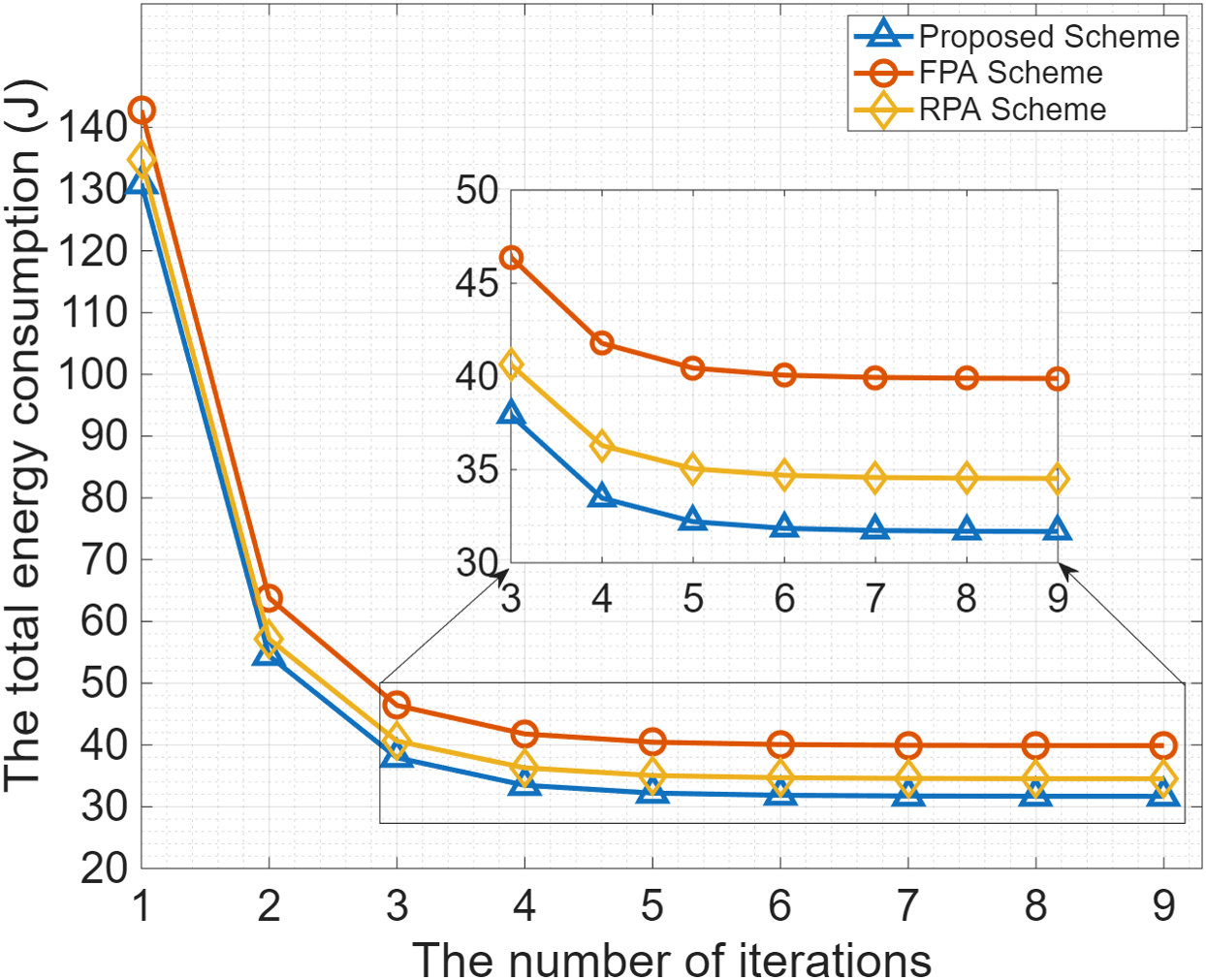}
            \caption{ The convergence performance.}
        \label{fig2}
\end{figure}

\begin{algorithm}[t]\footnotesize
\caption{AO-based Algorithm for Solving Problem (\textbf{P1}).}
\label{algorithm2} 
\begin{algorithmic}[1]
\REQUIRE
{An initial feasible solution $\{\mathbf{F}^{(0)},\mathbf{P}^{(0)},\mathbf{W}^{(0)},\mathbf{X}^{(0)},\}$, 
\STATE
        \textbf{Initialize:} iteration index $ l = 1 $, maximum iteration number $l_{\rm{max}}$, accuracy threshold $ \varepsilon > 0 $, iteration index of the PSO $i= 1$, and maximum iteration number of the PSO $i_{{\max}}$};\\
         Compute the initial objective function value of problem \textbf{P1}, i.e.$ \Phi^{(0)}(\mathbf{F}^{(0)},\mathbf{P}^{(0)},\mathbf{W}^{(0)},\mathbf{X}^{(0)})$;

\STATE \textbf{Repeat:}
\STATE Given $\{\mathbf{P}^{(l)},\mathbf{W}^{(l)},\mathbf{X}^{(l)}\}$, solve problem \textbf{P2} to obtain $\mathbf{F}^{(l+1)}$;
\STATE Given $\{\mathbf{F}^{(l+1)},\mathbf{X}^{(l)}\}$, solve problem \textbf{P3} by Algorithm \ref{algorithm1} to obtain $\mathbf{W}^{(l+1)}$ and $\mathbf{P}^{(l+1)}$;
\STATE Given $\{\mathbf{P}^{(l+1)},\mathbf{W}^{(l+1)},\mathbf{F}^{(l+1)}\}$, update $\mathbf{X}^{(l+1)}$ by the PSO; 
\STATE Update the objective function value of problem \textbf{P1};
\STATE Update $l = l + 1$;
\STATE \textbf{Until:} the decrease of the value of the objective function between two adjacent iterations is smaller than $\varepsilon$ or $l > l_{\max}$;
\ENSURE $\{\mathbf{F}^{(*)},\mathbf{P}^{(*)},\mathbf{W}^{(*)},\mathbf{X}^{(*)}\}$.
\end{algorithmic}
\label{algorithmAO}
\end{algorithm}

Fig. \ref{date} shows the total energy consumption versus the CE device task load under two antenna configurations, i.e., $K=6$ and $K=4$. As expected, the total energy consumption increases with the task load for all schemes, since heavier tasks require more communication and computation resources. Moreover, the schemes with $K=6$ antennas generally consume less energy than those with $K=4$ antennas, owing to the additional spatial degrees of freedom provided by larger antenna arrays. It can also be observed that the proposed MA-assisted scheme consistently outperforms the FPA and RPA benchmarks  under different task loads and antenna scales. This is because the proposed scheme can jointly optimize antenna positions and resource allocation, thereby effectively reducing energy consumption, especially under heavy task loads. In addition, the proposed scheme with only $K=4$ movable antennas achieves competitive or even better performance than the FPA scheme with $K=6$ fixed antennas, demonstrating the hardware efficiency and robustness of the MA architecture.
  
\begin{figure}
    \centering
    
    \includegraphics[width=0.605\linewidth]{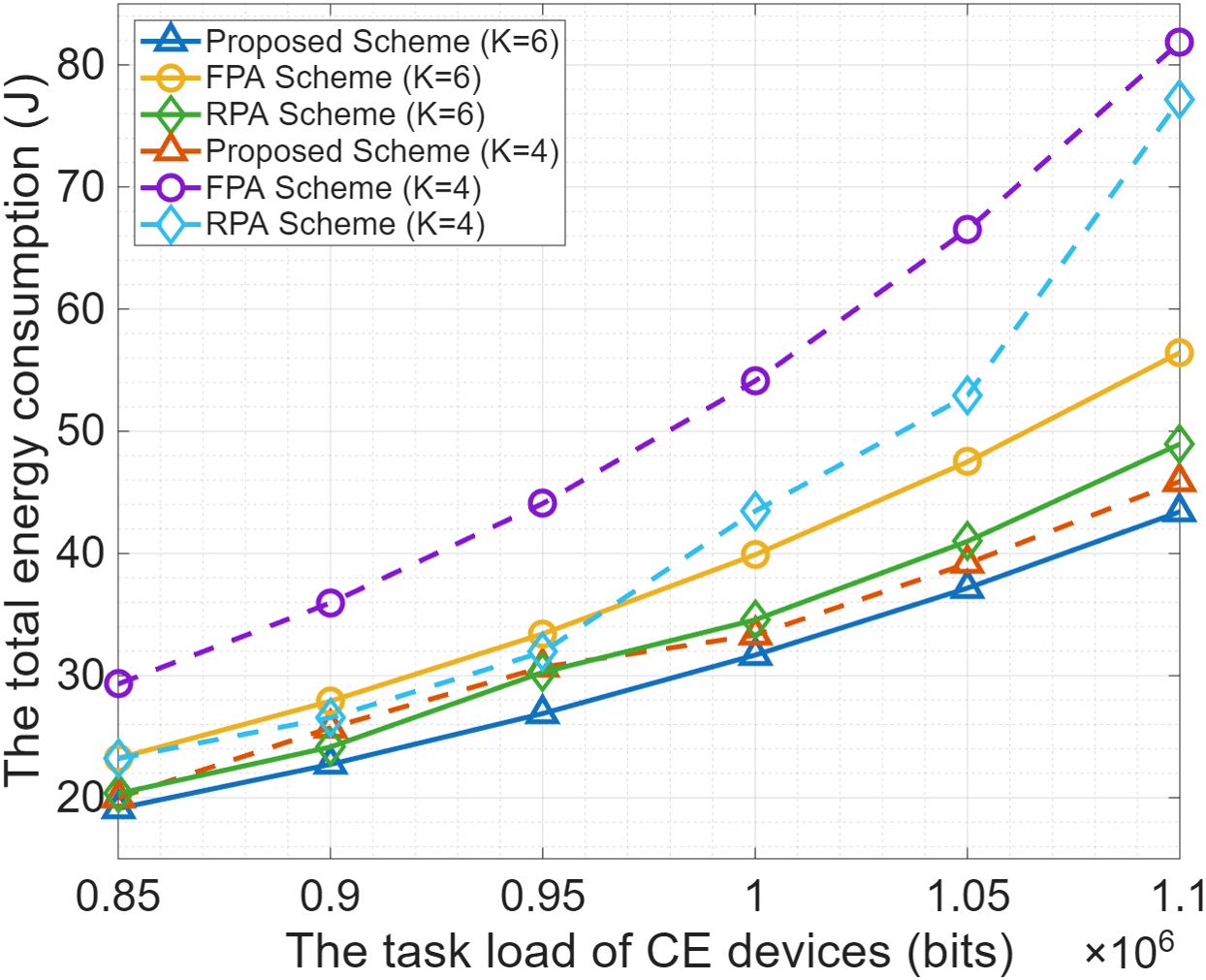}
           
           \caption{{Minimum total energy consumption versus the CE devices task load under different numbers of antennas.}}
        
           \label{date}
\end{figure}

The effect of computation intensity $c$ on total energy consumption is examined in Fig. \ref{device} under different numbers of antennas, i.e.,  $K=6$ and $K=4$.  The energy consumption of all schemes increases monotonically with $c$ due to higher CPU cycle demands. Owing to larger spatial DoFs and beamforming gains, the $K=6$ configuration consistently achieves lower energy consumption than $K=4$. Under both setups, the proposed MA scheme yields the lowest energy consumption, outperforming the FPA and RPA benchmarks, especially at high computation intensities. Notably, the proposed scheme with $K=4$ surpasses the FPA scheme with $K=6$, demonstrating the high hardware efficiency of MAs. This superiority stems from the MA's ability to optimize channels and boost the uplink rate $r_m$; the resulting reduction in offloading time $T_m^{\text{tran}}$ enables the UAV to operate at a more energy-efficient CPU frequency $f_m^{\text{uav}}$ without violating latency constraints.

\begin{figure}
    \centering
    \includegraphics[width=0.605\linewidth]{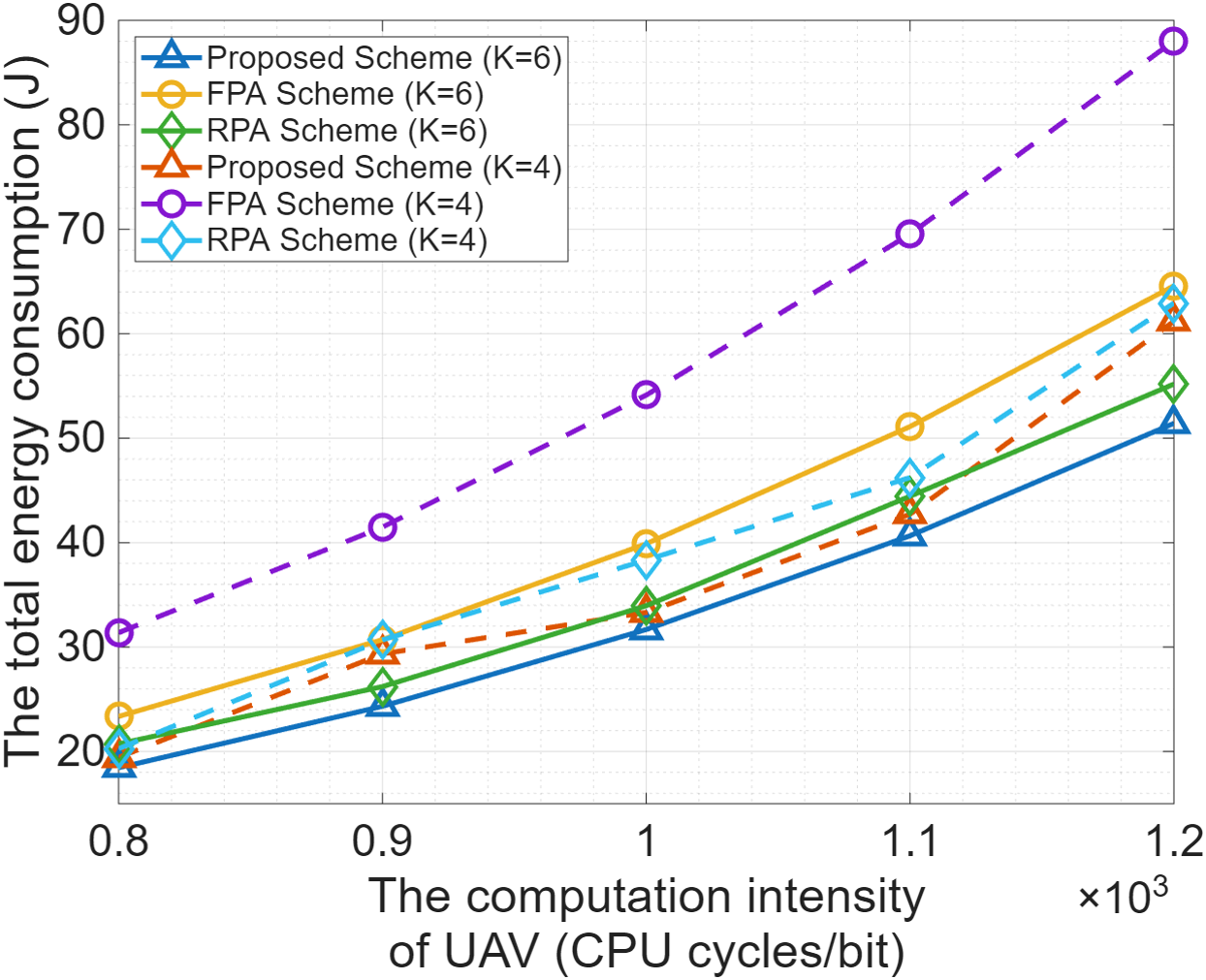}
           \caption{Minimum total energy consumption versus the computation intensity under different numbers of antennas.}
           
          \label{device}
\end{figure}

Fig. \ref{latency} compares the minimum total energy consumption under different maximum tolerable latency values for $K=6$ and $K=4$. It can be observed that the total energy consumption decreases as the latency constraint is relaxed, since the system can allocate lower transmit powers and computing resources over a longer duration. Moreover, increasing the number of antennas from $K=4$ to $K=6$ further reduces the energy consumption by providing additional spatial degrees of freedom. Under the same latency constraint and antenna configuration, the proposed MA-assisted scheme consistently achieves lower energy consumption than the FPA and RPA benchmark schemes, demonstrating that the energy saving is not obtained at the expense of users' QoS. More importantly, the proposed scheme with only $K=4$ movable antennas outperforms the FPA scheme with $K=6$ fixed antennas over the whole latency range. This confirms that adaptively adjusting antenna positions is an effective hardware-efficient approach to improving energy efficiency while satisfying stringent latency requirements.

\begin{figure}
    \centering
    \includegraphics[width=0.605\linewidth]{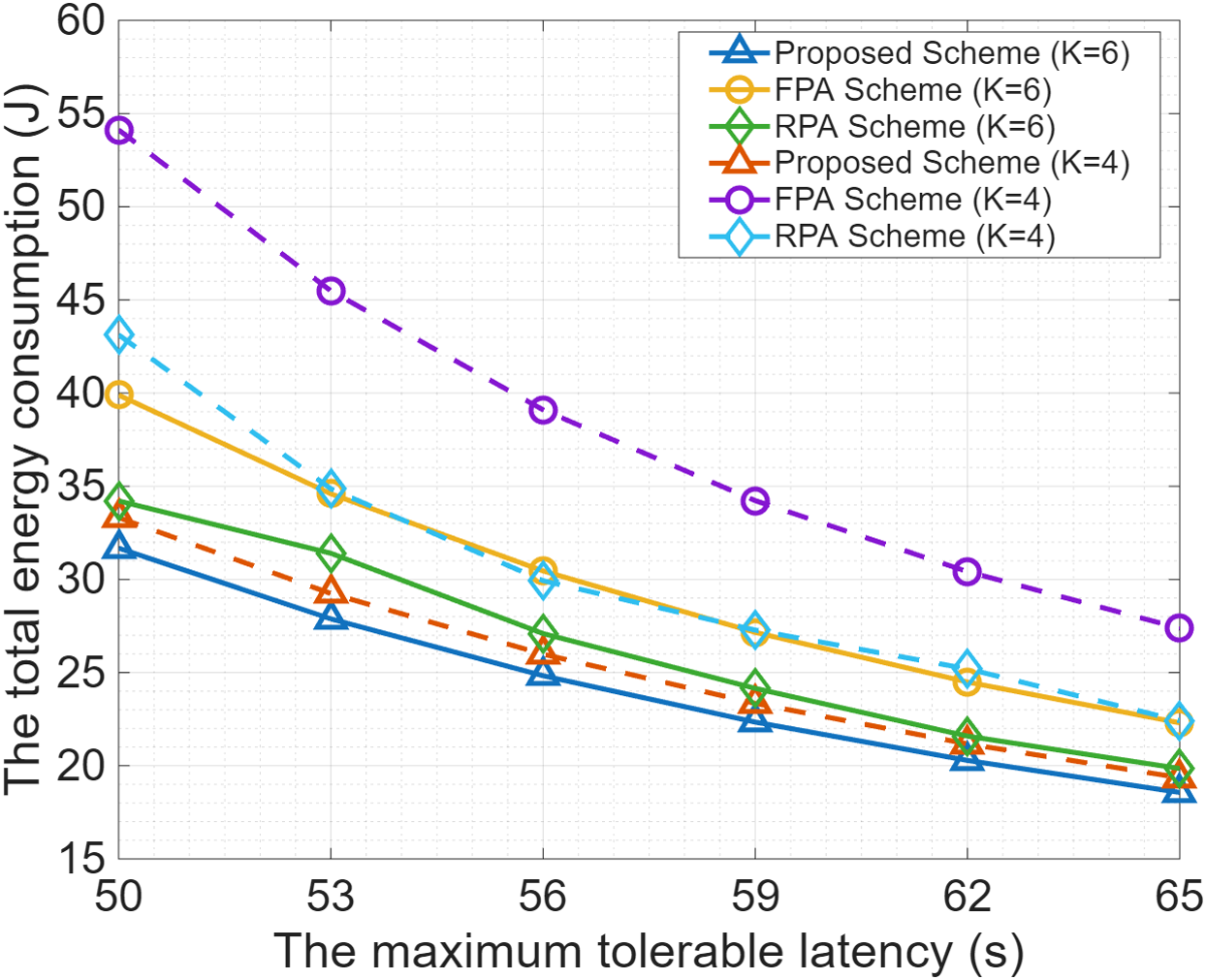}
            \caption{ Minimum total energy consumption versus the maximum tolerable latency under different numbers of antennas.}
            \label{latency}
\end{figure}

Having verified the hardware efficiency and QoS performance of the MA-assisted architecture, we next examine its scalability under dynamic traffic demands. Specifically, two task-volume scenarios are considered, including a heavy load of $1 \times 10^7$ bits and a light load of $0.85 \times 10^7$ bits.

Fig. \ref{fig3} illustrates the minimum total energy consumption versus the number of antennas under two task-volume scenarios, corresponding to a heavy load of $1 \times 10^7$ bits and a light load of $0.85 \times 10^7$ bits. Across both scenarios, energy consumption decreases monotonically as the number of antennas increases, since larger arrays offer enhanced spatial DoFs and superior beamforming gains to improve channel quality and reduce transmission duration. Furthermore, the proposed MA-assisted scheme consistently achieves the lowest energy footprint. This performance advantage over the FPA and RPA benchmark schemes become more pronounced at the higher task volume, demonstrating the critical role of antenna position optimization in energy-intensive scenarios. Unlike FPA scheme rigid channels and RPA scheme inconsistent random positioning, the proposed scheme adaptively reconfigures antenna positions for maximum efficiency.

 \begin{figure}
    \centering
    \includegraphics[width=0.605\linewidth]{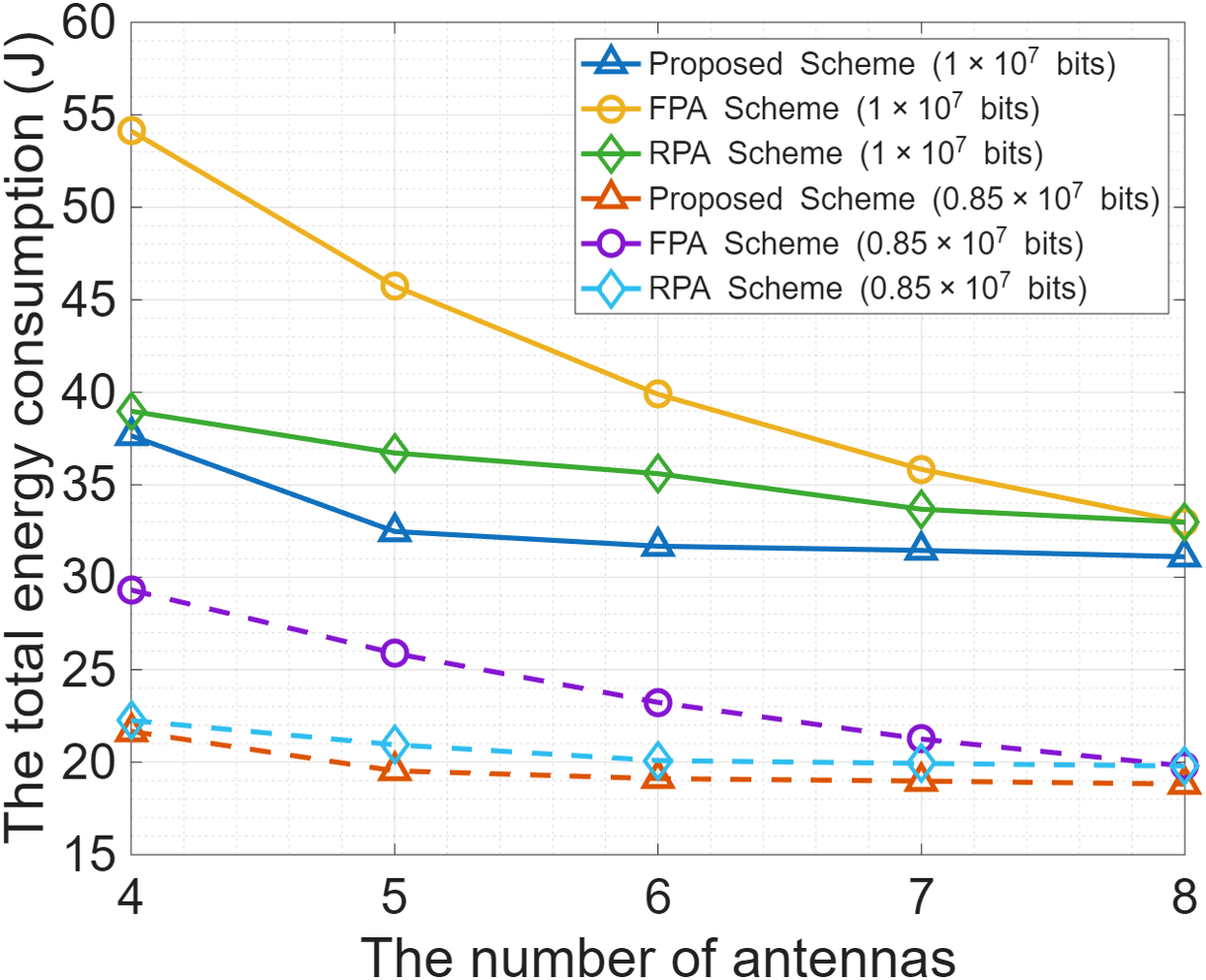}
           \caption{ {Minimum total energy consumption versus the number of antennas under different total task volumes.}}
           
           \label{fig3}
\end{figure}

The impact of the number of ground CE devices on total energy consumption is shown in Fig. \ref{users} for task volumes of $1 \times 10^7$ and $0.85 \times 10^7$ bits. Energy costs escalate across all schemes as the number of devices increases, driven by the aggregated workload and intensified multi-device interference. Naturally, heavier task loads incur higher energy consumption. Regardless of the configuration, the proposed MA-assisted scheme consistently achieves the lowest energy usage. Its performance advantage over the FPA and RPA benchmarks expands significantly as both the device density and task volume increase. This is particularly pronounced in the $1 \times 10^7$ bits scenario, where FPA scheme triggers a dramatic spike in energy consumption as devices must elevate transmit power to mitigate intense interference. In contrast, our proposed scheme intelligently reconfigures MA positions to establish favorable channel conditions. 
 
\begin{figure}
    \centering
    \includegraphics[width=0.605\linewidth]{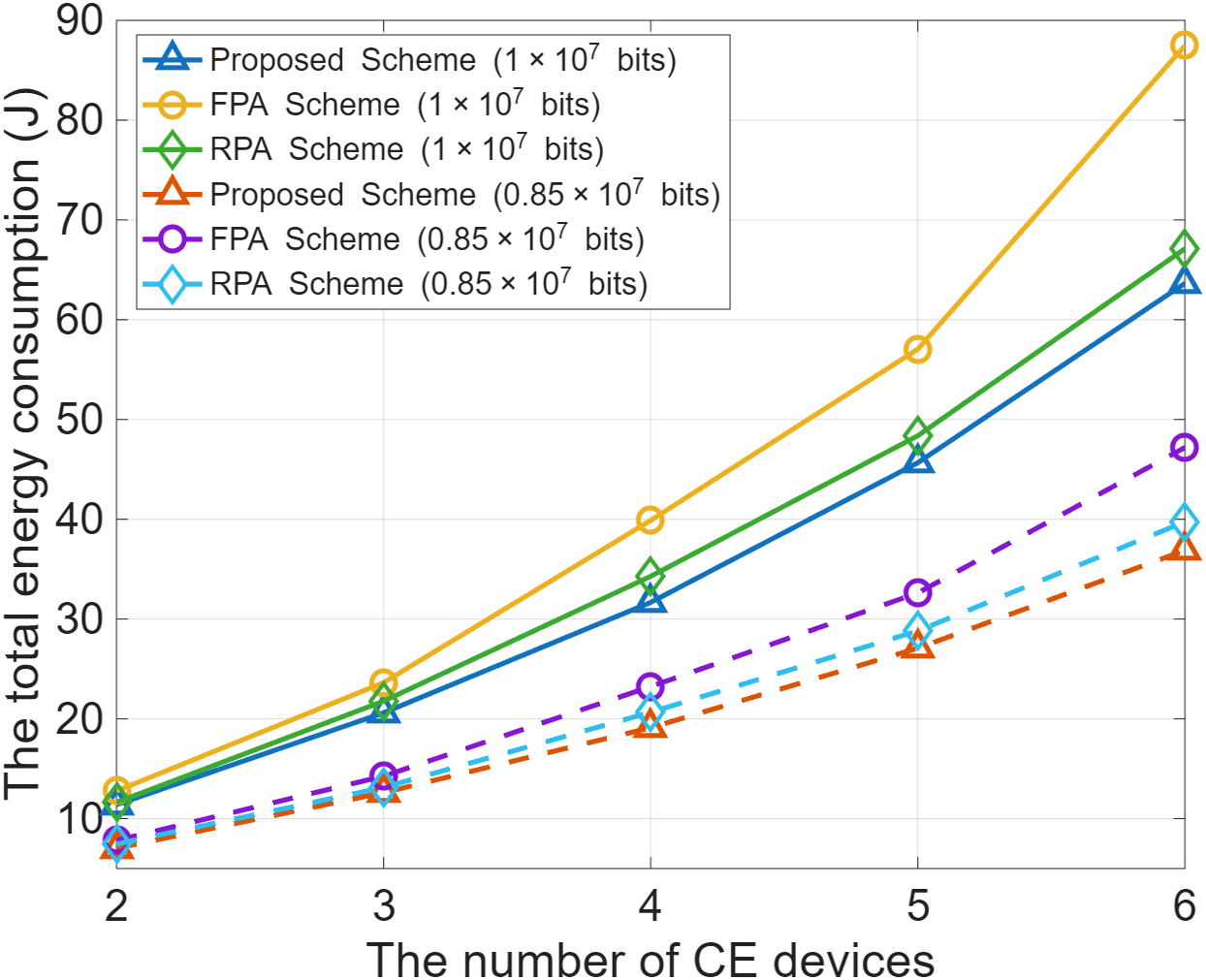}
             \caption{Minimum total energy consumption versus the number of CE devices under different total task volumes.}
             \label{users}
\end{figure}

Fig. \ref{bandwidth} presents the total energy consumption under different bandwidth settings for task volumes of $1 \times 10^7$ and $0.85 \times 10^7$ bits. For both workloads, energy consumption decreases monotonically as bandwidth expands, since a wider bandwidth elevates data rates and shortens transmission duration. Specifically, under the $1 \times 10^7$ bits load at 1 MHz, the proposed scheme reduces energy consumption by approximately 33.8\% against the FPA scheme and 15.7\% against the RPA scheme. This advantage stems from the flexible beamforming gains of MAs. By adaptively reconfiguring antenna positions, the proposed scheme enhances channel gains, thereby minimizing both CE device transmission energy and UAV operational computation energy consumption.

\begin{figure}
    \centering
  \includegraphics[width=0.605\linewidth]{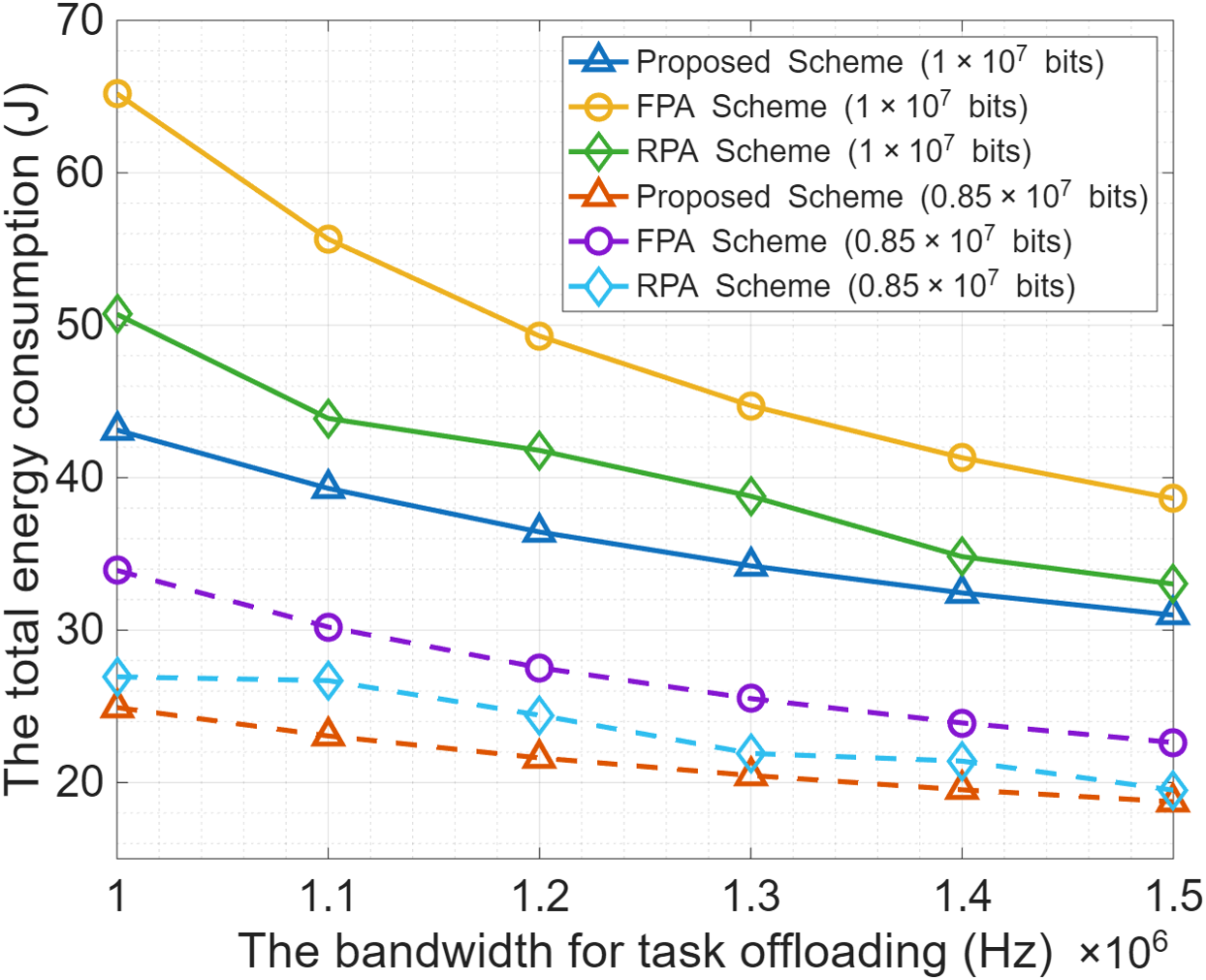}
         \caption{ { Minimum total energy consumption versus the bandwidth for task offloading under different total task volumes.}}
        \label{bandwidth}
\end{figure}

\section{Conclusion}
In this work, we investigated an MA-assisted UAV-enabled MEC system for CE networks. A UAV equipped with an MA array was employed as a flexible aerial edge server to support computation-intensive tasks. To minimize the total system energy consumption, we formulated a joint optimization problem involving  UAV computing resource allocation, CE devices' transmit power, receive beamforming, and MA positions, while considering communication, computation, and MA movement energy. An efficient AO algorithm was developed to solve the resulting non-convex problem, with complexity and convergence analyses provided. Simulation results demonstrated  the proposed MA-assisted scheme significantly reduces energy consumption compared with FPA scheme and RPA scheme. In particular, by exploiting the additional spatial flexibility of MAs, a compact MA array can achieve even better energy efficiency than a larger fixed-antenna array. These results verify the effectiveness and hardware efficiency of the proposed scheme for future consumer-oriented low-altitude networks.

\begin{appendices}
\section{Proof of Theorem \ref{the1}}
\label{App1}
\begin{proof}
The Lagrangian function associated with the user transmit power control problem $(\textbf{P4.2})$ is formulated as follows
\begingroup
\small
\setlength{\abovedisplayskip}{3pt}
\setlength{\belowdisplayskip}{3pt}
\setlength{\jot}{1pt}
\begin{equation}
\begin{aligned}
\varphi(\mathcal{X}_1)
&= \sum_{m=1}^{M}
\left(2\alpha_m\sqrt{l}p^{\mathrm{new}}_m
-\alpha_m^{2}(\tilde r_m)^{\mathrm{new}}\right) \\
&\quad - \sum_{m=1}^{M}
\varpi_m\left(\Upsilon_m-(\tilde r_m)^{\mathrm{new}}\right)
- \sum_{m=1}^{M}
\mathcal{Z}_{m}\left(p^{\mathrm{new}}_m-\sqrt{P_{\max}}\right).
\label{eq36}
\end{aligned}
\end{equation}
\endgroup
Let $\mathcal{X}_{1} {=} \{ \mathcal{Z}_m, p^{\mathrm{new}}_m , \varpi_m\}$ denote the set of optimization and dual variables, where $\{ \mathcal{Z}_m ,\forall m\}$ are the dual variables associated with constraint \eqref{22a}, $\{ \varpi_m ,\forall m\}$ are the dual variables associated with constraint \eqref{22b}, and $\Upsilon_m$ can be calculated as
\begingroup
\small
\setlength{\abovedisplayskip}{3pt}
\setlength{\belowdisplayskip}{3pt}
\setlength{\jot}{1pt}
\begin{equation}
\Upsilon_m =\frac{l}{(\tau-\max_{k \in K} \left\{ T^\text{mov}_k \right\}-\frac{lc}{ f^{\mathrm{uav}}_m } )}.
\end{equation}
\endgroup

 The partial derivatives of $ \varphi( \mathcal{X}_{1})$ with respect to $p^{\mathrm{new}}_m$ are given by:
 \begingroup
\small
\setlength{\abovedisplayskip}{3pt}
\setlength{\belowdisplayskip}{3pt}
\setlength{\jot}{1pt}
\begin{equation}
\begin{aligned}
 &\frac{\partial  \varphi( \mathcal{X}_1)}{\partial p_{m}^{\mathrm{new}}}  = 2\alpha_m\sqrt{l}-2By_m^{2}\mathcal{G}_m|\mathcal{I}_m|^2  \sum\limits_{j=1}^{ M}  p^{\mathrm{new}}_j \left| \boldsymbol{w}^H_m \boldsymbol{h}_j \right|^2\\
 &\quad +2B\varpi_m\mathcal{G}_m|\mathcal{I}_m|^2  \sum\limits_{j=1}^{ M}  p^{\mathrm{new}}_j \left| \boldsymbol{w}^H_m \boldsymbol{h}_j \right|^2 - \mathcal{Z}_{m} \\
  &\quad +2B\alpha_m^{2}\mathcal{G}_m\Re{\{\mathcal{I}_m^* \boldsymbol{w}^H_m \boldsymbol{h}_m\} } -2B\varpi_m\mathcal{G}_m\Re{\{\mathcal{I}_m^* \boldsymbol{w}^H_m \boldsymbol{h}_m\} }.
\end{aligned}
\end{equation}
\endgroup
According to the first-order optimality condition, the optimal solution to Problem $\mathbf{(P4.2)}$ is derived by setting the derivative of $\varphi(\mathcal{X}_{1})$ to zero, which yields the optimal transmit power as
\begingroup
\small
\setlength{\abovedisplayskip}{3pt}
\setlength{\belowdisplayskip}{3pt}
\setlength{\jot}{1pt}
\begin{align}
&(p_{m}^{\mathrm{new}})^{\mathrm{opt}} =\nonumber\\
&\frac{2\alpha_m\sqrt{l}- \mathcal{Z}_{m}+2B \mathcal{G}_m \Re{\{\mathcal{I}_m^* \boldsymbol{w}^H_m \boldsymbol{h}_m\} }(\alpha_m^{2}-\varpi_m)}{2B\mathcal{G}_m|\mathcal{I}_m|^2\sum_{j=1}^{M}\left| \boldsymbol{w}^H_m \boldsymbol{h}_j \right|^2(\alpha_m^{2}-\varpi_m)}.
\end{align}
\end{proof}
\endgroup

\section{Proof of Theorem \ref{the2}}
\label{App2}
\begin{proof}
The Lagrangian function associated with the CE device transmit power control problem $(\textbf{P4.3})$ is formulated as follows
\begingroup
\small
\setlength{\abovedisplayskip}{3pt}
\setlength{\belowdisplayskip}{3pt}
\setlength{\jot}{1pt}
\begin{equation}
\begin{aligned}
    \varphi( \mathcal{X}_2) & =\sum_{m=1}^{M} (2\alpha_m\sqrt{lp_m}-\alpha_m^{2}\tilde{ r}_m )\\
  &\quad - \sum_{m=1}^{M}\beta_m(\Upsilon_m-\tilde{ r}_m) - \sum_{m=1}^{M} \phi_{m}(\|\boldsymbol{w}_m\|^2  - 1).
\end{aligned}
\end{equation}
\endgroup
Let $\mathcal{X}_{2} {=} \{ \beta_m, \boldsymbol{w}_m, \phi_{m}\}$ denote the set of optimization and dual variables, where $\{ \beta_m,\forall m\}$ are the dual variables associated with constraint \eqref{eq26a}, $\{ \phi_{m},\forall m\}$ are the dual variables associated with constraint \eqref{eq11a}.

 The partial derivatives of $ \varphi( \mathcal{X}_{2})$ with respect to $\boldsymbol{w}_m$ are given by:
 \begingroup
\small
\setlength{\abovedisplayskip}{3pt}
\setlength{\belowdisplayskip}{3pt}
\setlength{\jot}{1pt}
\begin{equation}
\begin{aligned}
 \frac{\partial  \varphi( \mathcal{X}_{2})}{\partial \boldsymbol{w}_m} & = -\phi_{m}\boldsymbol{w}^*_m -B\alpha_m^{2}\mathcal{G}_m\sqrt{p_m }\boldsymbol{h}^*_m\mathcal{I}_m \\
 &+B\alpha_m^{2}\mathcal{G}_m|\mathcal{I}_m|^2  (\sum\limits_{j=1}^{ M} p_j  \boldsymbol{h}^*_m\boldsymbol{h}^{\text{T}}_m+\sigma^2\boldsymbol{I})\boldsymbol{w}^*_m  \\
  &+B\beta_m\mathcal{G}_m\sqrt{p_m }\boldsymbol{h}^*_m\mathcal{I}_m \\
  &-B\beta_m\mathcal{G}_m|\mathcal{I}_m|^2  (\sum\limits_{j=1}^{ M} p_j  \boldsymbol{h}^*_m\boldsymbol{h}^{\text{T}}_m+\sigma^2\boldsymbol{I})\boldsymbol{w}^*_m.
\end{aligned}
\end{equation}
\endgroup

According to the first-order optimality condition, the optimal solution to Problem $\mathbf{(P4.3)}$ is derived by setting the derivative of $\varphi(\mathcal{X}_{2})$ to zero, which yields the optimal transmit power as
\begingroup
\small
\setlength{\abovedisplayskip}{3pt}
\setlength{\belowdisplayskip}{3pt}
\setlength{\jot}{1pt}
\begin{align}
&(\boldsymbol{w}_m)^{\mathrm{opt}} =\nonumber\\
&\frac{(\alpha_m^{2}-\beta_m)B\mathcal{G}_m\sqrt{p_m }\boldsymbol{h}_m\mathcal{I}_m }{B\mathcal{G}_m|\mathcal{I}^*_m|^2  (\sum\limits_{j=1}^{ M} p_j  \boldsymbol{h}^*_m\boldsymbol{h}^{\text{T}}_m+\sigma^2\boldsymbol{I})(\beta_m-\alpha_m^{2})+\phi_{m}\boldsymbol{I}}.
\end{align}
\endgroup

\end{proof}

\section{Proof of Proposition \ref{prop1}}
\label{App3}
\begin{proof}
Let us define $\{\boldsymbol{F}^{(l)},\boldsymbol{P}^{(l)},\boldsymbol{W}^{(l)},\boldsymbol{X}^{(l)}\}$ as the variables updated in the $l$-th iteration, with $\Phi (\cdot)$ denoting the objective function of problem \textbf{P1}. Initially, given $\{\boldsymbol{P}^{(l)},\boldsymbol{W}^{(l)},\boldsymbol{X}^{(l)}\}$, we obtain the optimal computation allocation $\boldsymbol{F}^{(l+1)}$ from problem \textbf{P2}. This results in the following reduction in the objective value:
\begingroup
\small
\setlength{\abovedisplayskip}{3pt}
\setlength{\belowdisplayskip}{3pt}
\setlength{\jot}{1pt}
\begin{equation}
\begin{aligned}
&\Phi (\boldsymbol{F}^{(l)},\boldsymbol{P}^{(l)}, \boldsymbol{W}^{(l)}, \boldsymbol{X}^{(l)})\\
&\geq \Phi (\boldsymbol{F}^{(l+1)},\boldsymbol{P}^{(l)}, \boldsymbol{W}^{(l)},\boldsymbol{X}^{(l)}).\label{e36}
\end{aligned}
\end{equation}
\endgroup
Then, for fixed $\{\boldsymbol{F}^{(l+1)}, \boldsymbol{X}^{(l)}\}$, the transmit power $ \boldsymbol{P}^{(l+1)}$  and the receive beamforming $ \boldsymbol{W}^{(l+1)}$ are determined by solving problem \textbf{P3} via Algorithm \ref{algorithm1}, which implies
\begingroup
\small
\setlength{\abovedisplayskip}{3pt}
\setlength{\belowdisplayskip}{3pt}
\setlength{\jot}{1pt}
\begin{equation}
\begin{aligned}
&\Phi (\boldsymbol{F}^{(l+1)},\boldsymbol{P}^{(l)}, \boldsymbol{W}^{(l)}, \boldsymbol{X}^{(l)}) \\
&\geq \Phi (\boldsymbol{F}^{(l+1)},\boldsymbol{P}^{(l+1)}, \boldsymbol{W}^{(l+1)},\boldsymbol{X}^{(l)}).\label{e38}
\end{aligned}
\end{equation}
\endgroup
In the last step, the MA positions $ \boldsymbol{X}^{(l+1)}$ are refined by solving problem \textbf{P5} based on $\{\boldsymbol{F}^{(l+1)},\boldsymbol{P}^{(l+1)}, \boldsymbol{W}^{(l+1)}\}$. Consequently,
\begingroup
\small
\setlength{\abovedisplayskip}{3pt}
\setlength{\belowdisplayskip}{3pt}
\setlength{\jot}{1pt}
\begin{equation}
\begin{aligned}
&\Phi (\boldsymbol{F}^{(l+1)},\boldsymbol{P}^{(l+1)}, \boldsymbol{W}^{(l+1)}, \boldsymbol{X}^{(l)}) \\
&\geq \Phi (\boldsymbol{F}^{(l+1)},\boldsymbol{P}^{(l+1)}, \boldsymbol{W}^{(l+1)},\boldsymbol{X}^{(l+1)}).\label{e39}
\end{aligned}
\end{equation}
\endgroup
From \eqref{e36}-\eqref{e39}, the overall relationship is derived as:
\begingroup
\small
\setlength{\abovedisplayskip}{3pt}
\setlength{\belowdisplayskip}{3pt}
\setlength{\jot}{1pt}
\begin{equation}
\begin{aligned}
&\Phi (\boldsymbol{F}^{(l)},\boldsymbol{P}^{(l)}, \boldsymbol{W}^{(l)}, \boldsymbol{X}^{(l)}) \\
&\geq\Phi (\boldsymbol{F}^{(l+1)},\boldsymbol{P}^{(l+1)}, \boldsymbol{W}^{(l+1)},\boldsymbol{X}^{(l+1)}).\label{e40}
\end{aligned}
\end{equation}
\endgroup
Equation \eqref{e40} proves that the objective function is non-increasing. Since the constraint set of problem \textbf{P1} is strictly bounded, the convergence of Algorithm \ref{algorithmAO} is established.
\end{proof}

\end{appendices}

\bibliographystyle{IEEEtran}
\bibliography{myref}

\end{document}